\documentclass[11pt]{article}
\usepackage[left=1in,right=1in,top=1in,bottom=1in]{geometry}

\usepackage[normalem]{ulem} 
\usepackage[utf8]{inputenc} 
\usepackage[T1]{fontenc}    

\PassOptionsToPackage{unicode}{hyperref}
\usepackage{hyperref} 
\usepackage{url}            
\usepackage{booktabs}       
\usepackage{amsfonts}       
\usepackage{nicefrac}       
\usepackage{microtype}      
\usepackage[utf8]{inputenc}

\usepackage{xcolor,amsmath,amsfonts,amsthm}
\RequirePackage{fix-cm}
\usepackage{graphicx}
\newtheorem{theorem}{Theorem}

\newtheorem{lemma}{Lemma}

\usepackage{latexsym,amssymb,amsbsy, comment} 
\usepackage{bbm }
\usepackage{graphicx}
\graphicspath{{fig/}}
\usepackage{algorithm}
\usepackage{float}
\usepackage[normalem]{ulem}
\usepackage{wrapfig}
\usepackage{authblk}
\usepackage{tabularx}
\usepackage{makecell}
\usepackage{cite}
\usepackage{multirow}

\newcommand{\E}{\mathop{\mathbb{E}}}
\newcommand{\OPT}{\textsc{Opt}}
\DeclareMathOperator*{\argmin}{\arg\!\min}

\newcolumntype{C}[1]{>{\centering\arraybackslash\hspace{0pt}}p{#1}}
\newcolumntype{L}[1]{>{\raggedright\arraybackslash\hspace{0pt}}p{#1}}

\usepackage{enumerate}
\usepackage[customcolors,shade]{hf-tikz}
\usepackage{setspace, comment, hhline}
\usepackage[noend]{algpseudocode} 
\usepackage{endnotes,algorithm,dsfont}
\usepackage{breqn}
\let\theoremstyle\relax
\usepackage{float, multirow}
\usepackage{soul}
\usepackage{tabularx}
\usepackage{subcaption}
\usepackage{tablefootnote}
\algdef{SE}[SUBALG]{Indent}{EndIndent}{}{\algorithmicend\ }%
\algtext*{Indent}
\algtext*{EndIndent}

\allowdisplaybreaks

\allowdisplaybreaks 
\title{Balanced Districting on Grid Graphs with\\ Provable Compactness and Contiguity \vspace{15 pt}}
\author[1]{Cyrus Hettle}
\author[1]{Shixiang Zhu}
\author[1]{Swati Gupta}
\author[1]{Yao Xie}

\affil[1]{\small Georgia Institute of Technology \protect \\
{\small \tt\{chettle,\ shixiang.zhu,\ swatig\}@gatech.edu, yao.xie@isye.gatech.edu}}

\date{}

\begin{document}
\maketitle
\begin{abstract}
Given a graph $G = (V,E)$ with vertex weights $w(v)$ and a desired number of parts $k$, the goal in graph partitioning problems is to partition the vertex set V into parts $V_1,\ldots,V_k$. Metrics for compactness, contiguity, and balance of the parts $V_i$ are frequent objectives, with much existing literature focusing on compactness and balance. Revisiting an old method known as striping, we give the first polynomial-time algorithms with guaranteed contiguity and provable bicriteria approximations for compactness and balance for planar grid graphs. We consider several types of graph partitioning, including when vertex weights vary smoothly or are stochastic, reflecting concerns in various real-world instances. We show significant improvements in experiments for balancing workloads for the fire department and reducing over-policing using 911 call data from South Fulton, GA.
\end{abstract}

\section{Introduction}\label{sec:intro}

Given a weighted graph $G = (V,E)$ with vertex (or edge) weights $w: V \rightarrow \mathbb{R}$ and the desired number of parts $k$, partitioning the vertex set $V$ into $k$ parts $V_1, \hdots, V_k$ is one of the oldest and fundamental problems in operations research and statistics, with the first work dating back to the seminal work of Weaver and Hess, to the best of our knowledge~\cite{WeHe1963}. The solutions to this problem have created a huge impact in various applications ranging from drawing voting districts~\cite{validi2020imposing} and school districts~\cite{ca04schooldistricting} to police zones \cite{gass1968division}, clustering \cite{Mo1973}, and even divisions of large network problems for parallel computing \cite{gh02parallelcomputing}.

In operations research, the applicability of graph partitioning solutions in practice often depends on solution quality in terms of three criteria: 
\begin{itemize}
\item[(i)] {\it Contiguity:} In many real-world applications, being able to travel within each part efficiently is crucial. If each part is a beat patrolled by a single police officer or under the jurisdiction of a fire station, then an officer should not have to enter other beats to move between disparate parts of their own beat. Moreover, in applications such as drawing voting districts, contiguity may be legally necessary. Therefore, we require that each part of the partition should induce a contiguous region, or in other words the subgraph\footnote{By $G[V_i]$, we mean the graph induced by vertex set $V_i$, comprising only vertices in $V_i$ and the subset of edges in $E$ with both endpoints in $V_i$.} $G[V_i]$ should be connected for each $i \in [k]$.

\item[(ii)] {\it Balance:} If each part is to be assigned the same level of service, which could be a discrete unit such as a single school or a single political representative, overpopulated or overworked parts can become districts where populations are under-served relative to those in lower-weight parts. When balancing workloads among districts of a region, this becomes a fairness concern. In parallel computing applications, high-weight parts have longer evaluation times, slowing down the entire computation. Therefore, we require that the parts be balanced in the vertex weights, i.e., the total vertex weight in each part $V_i$ should be ``close" to the average weight of any of the $k$ parts. To ensure balance, the variance of total weights over the parts $V_i, \hdots, V_k$ can be minimized or constrained to be within $\varepsilon$ of the average desired weight. 

\item[(iii)] {\it Compactness:} As with contiguity, compactness can be motivated in terms of the travel time of each region. A highly oblong or non-convex zone can be time-consuming to traverse, and therefore increase workload in an application such as police districting. In parallel computing, ``compact" parts have few connections to their neighbors. Hence, if problems are solved on each part independently any errors that arise from solving a part without considering its surroundings will be minimized. Therefore, we desire that each part should be compact with respect to the underlying graph. This is often achieved by minimizing the size of the multi-cut (i.e., the number of edges with end-points in different parts of the partition) or minimizing the perimeter of the region (e.g., for planar graphs).
\end{itemize}

Many different formulations for partitioning problems exist, dependent on the modeling choice for each of these objectives. Most of the existing literature, however, focuses on ensuring compactness and balance~\cite{Bulu2016RecentAI,validi2020imposing}, but neglects contiguity of parts. Guaranteeing contiguity while ensuring tractability of districting over grids is the key contribution of our work. 

\subsection{Problem Formulation}\label{subsec:formulation} For applications such as police districting or fire workloads planning, one can superimpose a rectangular grid graph on a geographical region with vertex weights $w(v)$ that can be used to approximate the crime rate or workload in the corresponding square of the grid. For instance, Figure~\ref{fig:grid_dual} shows a ``map" divided into parts by the superimposed black grid. The dual graph of the map, minus the vertex representing the region outside the map, is shown in blue, resulting in the blue grid graph $G$.

\begin{figure}
    \centering
    \begin{subfigure}[h]{0.65\linewidth}
    \includegraphics[width=\linewidth]{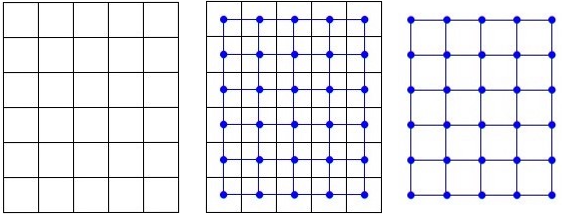}
    \end{subfigure}
    \begin{subfigure}[h]{0.2\linewidth}
    \includegraphics[width=\linewidth]{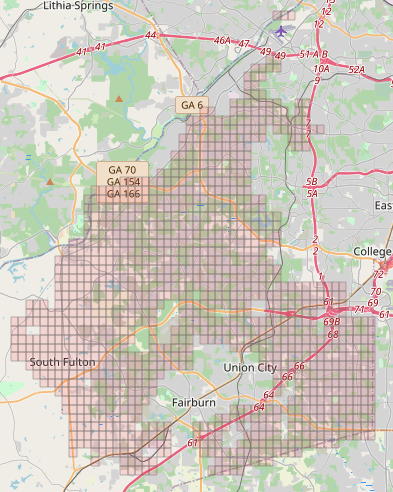}
    \end{subfigure}
    \caption{\scriptsize Left, a map with parts determined by a rectangular grid and its planar dual; Right, a square grid on the map of the city of South Fulton, GA.}
    \label{fig:grid_dual}
\end{figure}

As mentioned above, optimizing for balance of node weights in any part, maximizing compactness and ensuring contiguity of each part results in a multi-objective optimization problem. We convert this into a single-objective model, treating compactness as the objective and requiring contiguity and approximate balance of each part of the partition.  Multiple candidates for the ``correct" notion of compactness exist, and the most appropriate one may depend on the application. In our model, we consider the compactness of a partition to be the sum of the costs of the cut edges between each part. This definition allows us to prevent highly non-convex and oblong parts, while also being very well-suited to the graph-based model we consider.


Formally, the {\it Balanced Graph Partitioning} ({\sc Partition}) problem we consider is as follows: Given a planar graph $G=(V,E)$ with vertex weights $w: V \rightarrow \mathbb{R}$ and edge costs $c(\cdot): E \rightarrow \mathbb{R}$, parameters $k>0$, $\varepsilon\geq 0$, partition the vertex set $V$ into $k$ parts $V_1, \dots, V_k \subseteq V$ such that the each subgraph $G[V_i]$ induced by part $V_i$ is connected, total weight in each part is within $1\pm \varepsilon$ fraction of the average weight $A= (1/k) \sum_{v\in V}w(v)$ and the cost of the multi-cut induced by the partition is minimized: 
\begin{align}
\text{({\sc Partition}): ~~} \text{minimize}  ~~&\sum_{i=1}^k\delta(V_i) \quad \quad \quad\quad \quad \quad\quad\quad\quad \text{({\bf compactness})}\label{eq:compactness}\\
\text{subject to:} ~~& G[V_i] \text{ is connected, for each } 1\leq i\leq k, ~~ \nonumber \\&\quad \quad \quad \quad \quad\quad \quad\quad \quad\quad \quad \quad \quad \text{({\bf connectivity})}\label{eq:connectivity}\\
 & 1-\varepsilon \leq \frac{\sum_{v\in V_i}w(v)}{A} \leq 1+\varepsilon, \text{ for each } 1\leq i\leq k, ~~\nonumber\\ &\quad \quad\quad \quad\quad \quad\quad \quad\quad \quad\quad \quad\quad \quad\quad \text{({\bf balance})}\label{eq:balance}
\end{align}
where $\delta(V_i)$ denotes the number of cut edges, i.e.\ edges with one vertex in $V_i$ and the other vertex in $V_j, j\neq i$. We call this balance condition as {\it $(1\pm \varepsilon)$-balanced} and refer to the one-sided condition $\sum_{v\in V_i}w(v)/A \leq 1+\varepsilon$ as {\it $(1+\varepsilon)$-balanced.}

In many applications, vertex weights $w(\cdot)$ are often estimates based on historical data. These weights might represent the voting or school-age population, or a more complex measure of workload for police districting or political affiliation. In practice, however, these estimates may be inaccurate, and partitioning based on these weights may result in unintended consequences. If a districting is already in use, this districting itself may create errors in the estimates, particularly if the parts of the districting are unbalanced. For instance,~Lum and Isaac demonstrated that implementing predictive policing methods using machine learning models can create a feedback loop where the number of reported crimes increases in areas where historical data indicated a higher crime rate, perpetuating overpolicing in these areas \cite{lum2016predict}. In a districting model, if equal police resources have been allocated to parts with different crime rates, then their study suggests that reported crime rates may be disproportionately higher in the parts with a lower true crime rate. 

To model the above mentioned issues in historical data, we also consider stochastic instances where vertex weights are estimates based on historical data as follows.  For each $v\in V$, we have either a known vertex-dependent distribution $X_v$, or a set of independent random observations $w_{v, 1}, \dots, w_{v,t_v}$ drawn from an unknown vertex-dependent distribution $X_v$, where the distributions $X_v$ are not correlated. We refer to this as the \textit{stochastic weights model.} It may be used for applications where historical data may be biased or not perfectly representative, as well as those where weights change in an unknown way over time. We show guarantees on the \textit{expected} balance of partitions in the stochastic model. A particularly useful condition on the stochastic weights is a bound $\sigma^2 \leq c\mu(v)$ on the variance $\sigma^2$ and mean $\mu(v)$ of each random variable $X_v$, for some constant $c$. This assumption is easily satisfied by many random variables, such as Poisson distributions (possibly non-homogeneous).

\subsection{Contributions} Our key contribution is to obtain provable guarantees and computational tractability for planar graph instances while ensuring connectivity of regions. We give two approaches to produce districting plans in both the exact and stochastic weight models: combinatorial algorithms for the partition problem, revisiting an old striping method from 1996 \cite{christou1996optimal}, and a simulated annealing heuristic method. Our focus is on producing a good solution efficiently in the case where vertex weights $w(v)$ are not uniform, reflecting many real-world instances. Our contributions include the following:

\begin{itemize}
    \item[(i)]  We give a modified striping method for unweighted rectangular grid graphs which guarantees contiguity. It improves the previous known multiplicative approximation factor for this instance by an asymptotic factor of 6 when considering perimeter as compactness (and guaranteeing contiguity), and we extend the multiplicative approximation to a much larger family of instances by removing conditions on the grid size and number of parts.
    
    \item[(ii)] We propose a new dynamic programming adaptation of the striping method, for partitioning general weighted planar grid graphs. This results in a fast, easily described algorithm suited for a range of applications. The guarantees for this method depend on the striping path chosen.  
    
    \item[(iii)] We further combine ideas from stochastic load balancing with combinatorial striping methods to provide an improvement in expected balance in the stochastic weights setting. This helps us reduce over-policing, for instance, with the real-world 911 calls data from South Fulton City.
    
    \item[(iv)] Finally, we compare our methods on several real-world and synthetic instances, including redistricting problems for the fire department and on police beats in South Fulton City, Georgia. To get tractable solutions in practice, we combine the striping approaches with simulated annealing to balance weights while maintaining contiguity and compactness. Our case studies exemplify the importance of scalable approaches to larger geographic regions while satisfying balance conditions, following the work of \cite{zhu2020data} where one of our initial plans based on simulated annealing on the same data was implemented in South Fulton City. 
\end{itemize}

The rest of the paper is organized as follows: in Section \ref{sec:related} we discuss related work, in Section \ref{sec:combinatorial} we give novel combinatorial striping algorithms for unweighted graphs, dynamic programming approach for general weights, and develop weight proxies to obtain a balanced partition under stochastic vertex weights. We present computational experiments on synthetic and real-world case studies from South Fulton County using 911 calls data, in Section~\ref{sec:computations}.

\section{Related Work}\label{sec:related} 

There has been a significant amount of work on balanced graph partitioning, often {\it without} the additional constraint on connectivity for each part. Police zone districting is a classical application of graph partitioning methods in operations research \cite{larson1972urban}. The districting problem has also been studied in other settings such as school districting and political districting \cite{GaNe1970}.
These problems are related to extensive literature on graph partitioning, which is known to be NP-complete, even for special cases such as unweighted grid graphs~\cite{FELDMANN201361}. 

If we relax any one of the three conditions out of balance, compactness and contiguity, we get a substantially different problem. Relaxing contiguity, we obtain the $k$-balanced minimum partition problem, a generalization of the NP-hard minimum bisection problem, for which both approximation algorithms and more practical heuristics are well-studied~\cite{andreev2006balanced,krauthgamer2009partitioning,diekmann2000shape}. Relaxing balance gives an upper bound but no lower bound on the size of each part and the number of parts is not fixed, and therefore, we obtain the $\rho$-separator problem\footnote{Given a graph $G=(V,E)$ with vertex weights $w(v)$ for all $v$ and a parameter $0< \rho<1$, the $\rho$-separator problem is to find a minimum cut $C\subseteq E$ such that each connected component of $G\setminus C$ has total weight at most $\rho$ fraction of the total weight of $G$.}, for which an approximation algorithm was given by Even et al. \cite{even1999fast}. Relaxing both balance and contiguity, we obtain the minimum $k$-cut problem, which is still NP-hard but for which good approximations exist~\cite{li2019faster}. Finally, if we do not require compactness, we obtain problems related to gerrymandering, where balance and contiguity may be mandated but objectives related to the distribution of votes, and particularly to the number of parts with a majority of votes from each political party, take precedence over compactness~\cite{duchin2018gerrymandering}. Feasibility problems in this setting are in general NP-hard, although they are tractable on some instances such as Hamiltonian graphs~\cite{APOLLONIO20093601}.

\begin{table}[t]

    \centering
    \resizebox{\textwidth}{!}{%
    \begin{tabular}{|c | c | c | c | c|}
    \hline
     {\bf Related Work} & {\bf Setting} & {\bf Balance} & {\bf Compactness} & {\bf Cont.}\\\hline\hline
   Rec. decomposition \& recombination \cite{andreev2006balanced} (polynomial time) & Polynomially weighted G& Param. & $O(\log^2 n)$ & No \\\hline
     Rec. bisection \cite{feldmannthesis} (polynomial time) & Grid graphs & 2-approx & $O(\log n)$& No \\\hline
         Rec. bisection \cite{ST} (polynomial time)& Well-shaped mesh*  & Exact only & $O(1)$& No \\\hline
    SDP rounding \cite{krauthgamer2009partitioning} (polynomial time) & General weighted  & 2-approx & $O(\sqrt{\log n\log k})$ & No \\\hline
    Striping \cite{christou1996optimal} (polynomial time) & Rectangular grid & Exact only & ($1+9/\lceil2 \sqrt{A}\rceil)$& No \\\hline\hline
    IP w/ flow \cite{Shirabe2009} & General graph & Param. & Exact & Yes\\\hline
    IP w/ separators \cite{oehrlein2017cutting} & General graph & Param. & Exact   & Yes \\\hline
    Modified $k$-means \cite{Mo1973} & Vertices in a plane & Exact & none & Yes  \\\hline
    Sim. annealing \cite{Damico2002} & General graph & Param. & none & Yes  \\\hline   
    $k$-means++ \cite{wei2016constant} & Vertices in a plane & none & $O(1)^\ddag$ & Yes$^\ddag$\\\hline \hline 
   {\bf Modified striping [this work] (polynomial time)} & Rectangular grid & {\bf Param.} & {\bf Theorem~\ref{thm:exactunweightedstriping}}& {\bf Yes}  \\\hline
    \end{tabular}
    }
    \caption{\scriptsize Summary of related work: compactness is $\sum_{i=1}^{k}\delta(V_i)$; $n = |V(G)|$; $k$ is the required number of parts; $A$ is the average weight of any part. Parameterized balance (abbreviated param.) is when the total weight of each part is between $(1-\varepsilon)A$ and $(1+\varepsilon)A$ for arbitrary $\varepsilon$. All approximation factors are multiplicative. *A well-shaped mesh is a planar graph embedding that has bounded aspect ratio (the ratio between its diameter and volume) and has angles that are not too small~\cite{ST}  $^\ddag O(1)$ for sum of squared distance metric which minimizes distance to chosen cluster centers; contiguity is respected when the embedding is convex.} 
\label{table:guarantees}
\end{table}

Feige and Krauthgamer made one of the first breakthroughs on the theoretical problem with no contiguity requirement, with a polylogarithmic approximation factor for the bisection (not $k$-partition) problem, with the objective being to minimize the number of cut edges~\cite{FK}. This improved over the previous best bound of $\tilde{O}(\sqrt{n})$. Moreover, their algorithm achieved an $O(\log n)$ approximation factor for planar graphs. Andreev and R\"{a}cke obtained a bicriteria PTAS, with an $O(\log^{1.5}n/\varepsilon^2)$ approximation factor for the number of cut edges of the partition \cite{andreev2006balanced}. Using SDP rounding techniques, Krauthgamer et al.\ improved this approximation factor to $O(\sqrt{\log n \log k})$~\cite{krauthgamer2009partitioning}. While these algorithms may yield contiguous regions for some specific instances, their methods give no guarantees on contiguity. Furthermore, they often employ repeated calls to SDP solvers, which can be time consuming. Partitioning problems are also closely related to the well-known algorithms for generating separators, an area which has been studied for decades beginning with the work of Lipton and Tarjan~\cite{LT}; however, these also do not give guarantees on contiguity \cite{survey}. 

A frequently used heuristic for partitioning or clustering algorithms is Lloyd's algorithm, or the $k$-means approach. This iterative approach is efficient and produces contiguous regions. Furthermore, modifications such as the $k$-means++ algorithm have provable approximation guarantees~\cite{arthur2007,wei2016constant}. However, these algorithms may not guarantee balance and they employ the compactness measure
$\sum_{v\in V} \min_{c\in C}||v-c||^2$,
where $C$ is a set of centers, corresponding to the parts of the partition. This choice of compactness metric differs from the one we choose and induces different behavior. For instance, if the density of points varies greatly across the map, an algorithm using the $k$-means objective will tend to produce more uniformly-sized regions than the perimeter objective used in our model. Alternatively, a modified iterated $k$-means approach can guarantee balance and contiguity, but lacks the compactness guarantees of $k$-means++ methods~\cite{Mo1973, Mo1976}.

There is also a large body of work on integer programming methods for various types of geographical districting problems, including the pioneering work of Garfinkel and Nemhauser~\cite{GaNe1970}. IP methods often handle balance and compactness well, but contiguity constraints can be a computational bottleneck. A variety of formulations to impose contiguity in IP models for districting have been proposed, including tree search with enclave elimination~\cite{Gr1985}, constraints that simultaneously impose contiguity and restrict compactness~\cite{Me1972}, the hub-flow constraints of~\cite{shirabe2009districting}, and the graph separator cutting plane methods of~\cite{oehrlein2017cutting}. Validi et al.\ give a stronger IP formulation based on a modification of the graph separator method exploiting the balance constraints, with a focus on tractability~\cite{validi2020imposing}.

Although in this work, we focus on 911 calls data (i.e., police department, fire department workloads), our methods are general. We include a discussion of related work specific to police districting and fire departments balancing in the appendix. Other methods for redistricting apply meta-heuristics, e.g., local search~\cite{WeHe1963}, genetic algorithms, simulated annealing~\cite{Damico2002} to geographical districting, which usually lack optimality or approximation guarantees but can guarantee contiguity. In particular, methods using simulated annealing, including
\cite{kirkpatrick1983optimization, vcerny1985thermodynamical, aarts1987simulated, van1987simulated, johnson1989optimization, d2002simulated, kim2011optimization} have been widely adopted recently because of its computational efficiency and ability to overcome trapping in local optima. 

\section{Combinatorial Striping Algorithms} \label{sec:combinatorial}

We now consider combinatorial algorithms for the districting problem. First, in this section we focus on the special case where $G$ is a grid graph, with $m$ rows and $n$ columns, and where edge costs and vertex weights are uniform. We modify the striping technique of Christou and Meyer in~\cite{christou1996optimal} to ensure contiguity and show improved approximation bounds by a factor of at least 6.

In many applications, one can represent a continuous distribution of weights by a grid of an appropriately chosen mesh size (e.g., Section~\ref{sec:case_study}). In this setting, the feasible shapes for a single district are limited to the polyominoes, shapes comprising a number of unit squares, attached to each other along their edges of a certain size (e.g.\ Figure~\ref{fig:striping15}), whose perimeter is easy to describe. 

\subsection{Striping for Uniform Weights}\label{subsec:uniformstriping}
To find balanced partitions of grid graphs, we construct specific combinatorial {\it striping} variations. We divide the vertical axis of the grid into stripes, possibly of varying heights. Within each stripe, we then create parts $B_i$ one-by-one by moving left to right along the columns and then down each column, accruing vertices as we do so, until the desired number $A$ vertices have been accrued. We collect those $A$ vertices into a part and continue in the same way. We call this ``top-to-bottom" striping. (See Figure~\ref{fig:striping15} for an example.) If the figure is rotated 90$^\circ$ clockwise, we get a striping pattern that we refer to as ``right-to-left" striping.

\begin{figure}[t]
\centering
  \includegraphics[width=.7\textwidth]{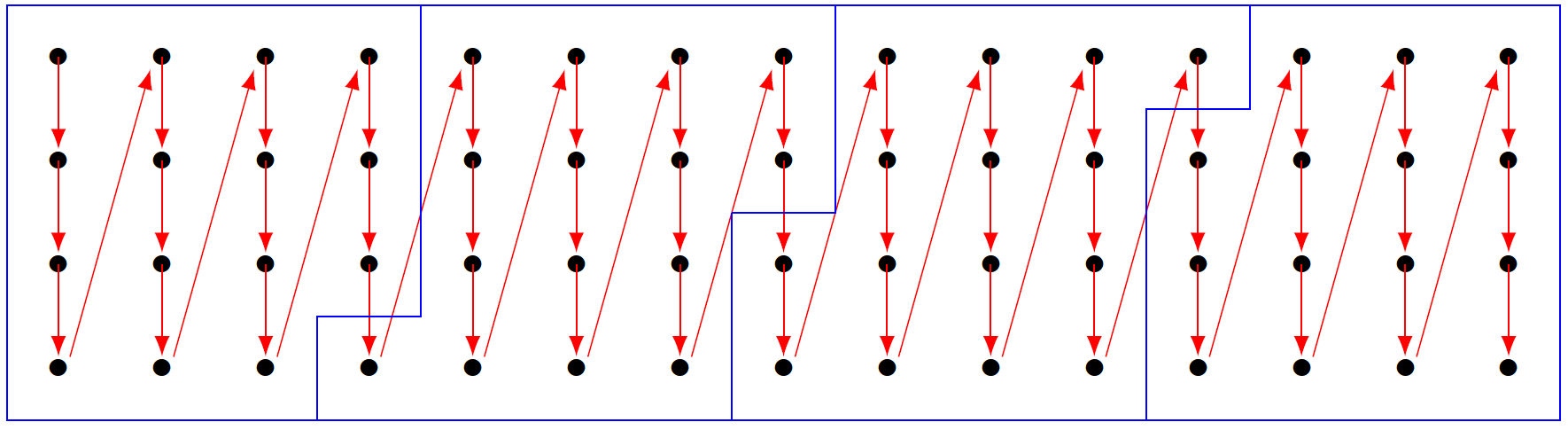}
 \caption{Example of striping, dividing a $4\times 15$ stripe into parts of size 15.}
\label{fig:striping15}
\end{figure}

Observe that the minimum-cut shapes in the square grid with a given number of vertices $A$ are square-like. Indeed, in general, they are given by taking a square or nearly-square rectangle, and if necessary, deleting vertices from the corners to obtain the correct size. If they contain no vertices on the border of $G$, their cut size is $2\lceil 2\sqrt{A}\rceil$~\cite{christou1996optimal}. For instance, there are two optimal shapes with 15 vertices: either a $4\times 4$ square with one corner vertex removed, or an intact $3\times 5$ rectangle. Both have 16 cut edges. If we choose a stripe height close to $\sqrt{A}$, the striping method does a near-optimal job of producing parts with shape close to these ideals. In Figure~\ref{fig:striping15}, with $A=15$ and a chosen a stripe height of 4, the first and fourth parts achieve this minimum cut, while the second and third have cut 18.

In various extensions of Christou and Meyer's original work, the striping technique has been studied for grid graphs with uniform weights and edge costs in several papers. For instance, in~\cite{christou1996optimal} Christou and Meyer give an algorithm for general $m,n,$ and $k$ (with $k\geq m,n$ and $A= mn/k\in\mathbb{N}$) with relative approximation error ${9}/{\lceil 2 A\rceil}$ (i.e., the total perimeter, including border edges of the rectangle, achieved by their solution is at most $1+{9}/{\lceil 2 A\rceil}$ times that of the optimum), but this algorithm does not guarantee contiguity. They employ stripes of height $\lfloor{\sqrt{A}}\rfloor$ and  $\lfloor{\sqrt{A}}\rfloor+1$, producing nearly square regions. These methods can be improved and expanded, for instance in the dynamic striping heights for irregular regions employed by Donaldson and Meyer in~\cite{Meyer2000ADH}, but without provable approximation guarantees compared to the global optimum.


\noindent
Finding minimum-cut balanced partitions on solid grid graphs is a special case of the $k$-balanced minimum partition problem, but this is still NP-hard. In fact, Feldmann showed that unless $\mathrm{P}=\mathrm{NP}$, for any $c>1/2$ and any $c>0$ there is no polynomial-time $(1+n^c/\varepsilon^d)$-approximation algorithm for the problem of finding an $\varepsilon$-balanced partition with the objective of minimizing the number of cut edges, using a reduction from the {\sc 3-Partition} problem~\cite{FELDMANN201361}.\footnote{Given a multiset $S$ of $n = 3 m$ positive integers, the {\sc 3-Partition} problem is to determine whether there is a partition of $S$ into $m$ subsets each of cardinality 3, such that each subset has the same sum.}

\subsection{\texorpdfstring{$\varphi$}{ϕ}-Cautious Striping Algorithm} Our novel striping algorithm, {\sc $\varphi$-Cautious  Striping}, {\it guarantees contiguous} regions, and for rectangular grid graphs achieves a relative multiplicative approximation error bound that is nearly six times better than previously known guarantees for perimeter. The key idea of our modified algorithm is to employ the striping technique as shown in Figure~\ref{fig:striping15} on strips of height approximately the square root of the number of vertices in each part, partitioning most of each strip but leaving approximately one part's worth of vertices at the end. This contrasts with previous methods, where the width of pieces remaining in each stripe is smaller and less critical to the algorithm. The leftover parts of each strip constitute an approximately rectangular region, which is then partitioned using a similar striping method which travels from row to row, top to bottom. We use $A=\lfloor mn/k \rfloor$ to denote the minimum number of vertices in a balanced part. This may be slightly different from the definition above of $A$ as the average weight of a part. However, note that in this unweighted context, we deal only with partitions that are as balanced as possible (hence having $A$ or $A+1$ vertices) and so do not need to consider the average weight of a part. We describe the algorithm in detail below.
\begin{algorithm}[!t]\footnotesize
\caption{{\sc $\varphi$-Cautious Striping}}
 \begin{algorithmic}
\State \textbf{Input:} A rectangular grid graph $G=(V,E)$ with $m$ rows, $n$ columns ($m <n$);  desired number of parts $k$; $A=\lfloor mn/k \rfloor$: the minimum number of vertices in a balanced part.
\State \textbf{Output:} A partition of $V$ into $k$ connected parts, with each part having $A$ or $A+1$ vertices (i.e., balanced), and total cut size at most $\OPT$, where $\OPT$ is the minimum total cut over all such partitions.
\State Let $m=da+r$ for some $d\in \mathbb{N}, 0\leq r\leq a-1$, where $a=\lfloor \sqrt{A} \rfloor$ is the minimum side length of all rectangles with minimum cut size that contain at least $A$ vertices.
\State {\bf Step 1. Divide the $m$ rows into horizontal strips of height:}
\Indent
\State $a$ and $a+1$, if $r \leq d$; $a$, and one of height $r$, if $r > \lfloor  a/\phi\rfloor$; $a$, and one of height $a+r$ if $r \leq \lfloor a/\phi\rfloor$.
\State Let strips of height $a$ and $a+1$ be in set $S_1$, and the other strip (if exists) in $S_2$.
\EndIndent

\State {\bf Step 2. Divide $S_1$ vertically until the number of unpartitioned columns in each strip is in $[\lfloor a/\phi\rfloor ,\lceil \phi a \rceil ]$:}
\Indent 
\State Each part within a strip is obtained using top-to-bottom striping and has area $A$ or $A+1$, while ensuring that the total number of parts with area $A+1$ is exactly $mn \mod k$.
\EndIndent 

\State {\bf Step 3. Divide $R$, which is the set of unpartitioned vertices in $S_1$, horizontally:}
\Indent 
\State Each part in $R$ is obtained using right-to-left and left-to-right striping and has area $A$ or $A+1$, while ensuring that the number of parts with area $A+1$ is exactly $mn \mod k$. (The last part may have fewer than $A$ vertices after this step.)
\EndIndent 

\State {\bf Step 4. Divide $S_2$ vertically:}
\Indent 
\State (a) If one part in $R$ has fewer than $A$ vertices, complete it by adding vertices in $S_2$ using top-to-bottom striping until it has area $A$ or $A+1$.
\State (b) Other parts in $S_2$ are obtained using top-to-bottom or bottom-to-top striping and have area $A$ or $A+1$, so that exactly $mn \mod k$ parts have area $A+1$.
\EndIndent
\State{\Return{the resulting partition $\mathcal{P}$.}}
\end{algorithmic}\label{alg:unweighted_striping}
\end{algorithm}

Suppose the graph has $m$ rows and $n$ columns and the desired number of parts is $k \geq n\geq m$, so that each part has $A= \lfloor mn/k \rfloor$ or $A+1$ vertices\footnote{Note that in the context of this algorithm, we use $A$ to denote the (integral) number of vertices in the smallest part, not the (possibly fractional) average number of vertices in a part.}. First, divide the $m$ rows of the graph into strips of height $a$ and $a+1$, where $a = \lfloor \sqrt{A} \rfloor.$ If this is not possible (which it may be if $m< a(a-1)$) then divide the $m$ rows into strips of height $a$, forming a rectangle $S_1$, with a strip $S_2$ of height approximately $a$ left over at the bottom. Using the striping method described above, we partition each of the strips of height $a$ or $a+1$ into parts of size $A$, stopping when between $\lfloor (\phi-1)a \rfloor$ and $\lceil \phi a \rceil$ complete columns are left in the strip. After completing this process on each strip, we are left with an approximately rectangular region $R$ on the right-hand side of the grid and an overlapping rectangular region $S_2$ at the bottom of the grid. We then apply the striping process to $R$, traveling top-to-bottom down the rows. Since the choice of stopping point in the previous striping steps ensures that $R$ has width approximately $a$ throughout, the resulting parts are not too far from being square and can be shown to have bounded cut size. We then continue by striping $S_2$, traveling right-to-left along the columns. Figure~\ref{fig:phistriping} shows the results of applying this algorithm to a $10\times17$ grid graph, with $k=10$. We include psuedocode for the algorithm in Algorithm~\ref{alg:unweighted_striping}.

The proof upper-bounds the number of cut edges created by each step of the algorithm, and then bounds their sum in terms of $m,n,A,$ and $d$. Likewise, \OPT\ is lower-bounded in terms of $m,n,A,$ and $d$. The non-constant terms in the ratio between these bounds include the expressions $(A+1)/A$ and $d/(d-2)$, which approach 1 in the limit. Evaluating the bound for minimum values of $A$ and $d$ gives the constant bound.
\begin{figure}[!t]
\centering
  \includegraphics[width=.46\textwidth]{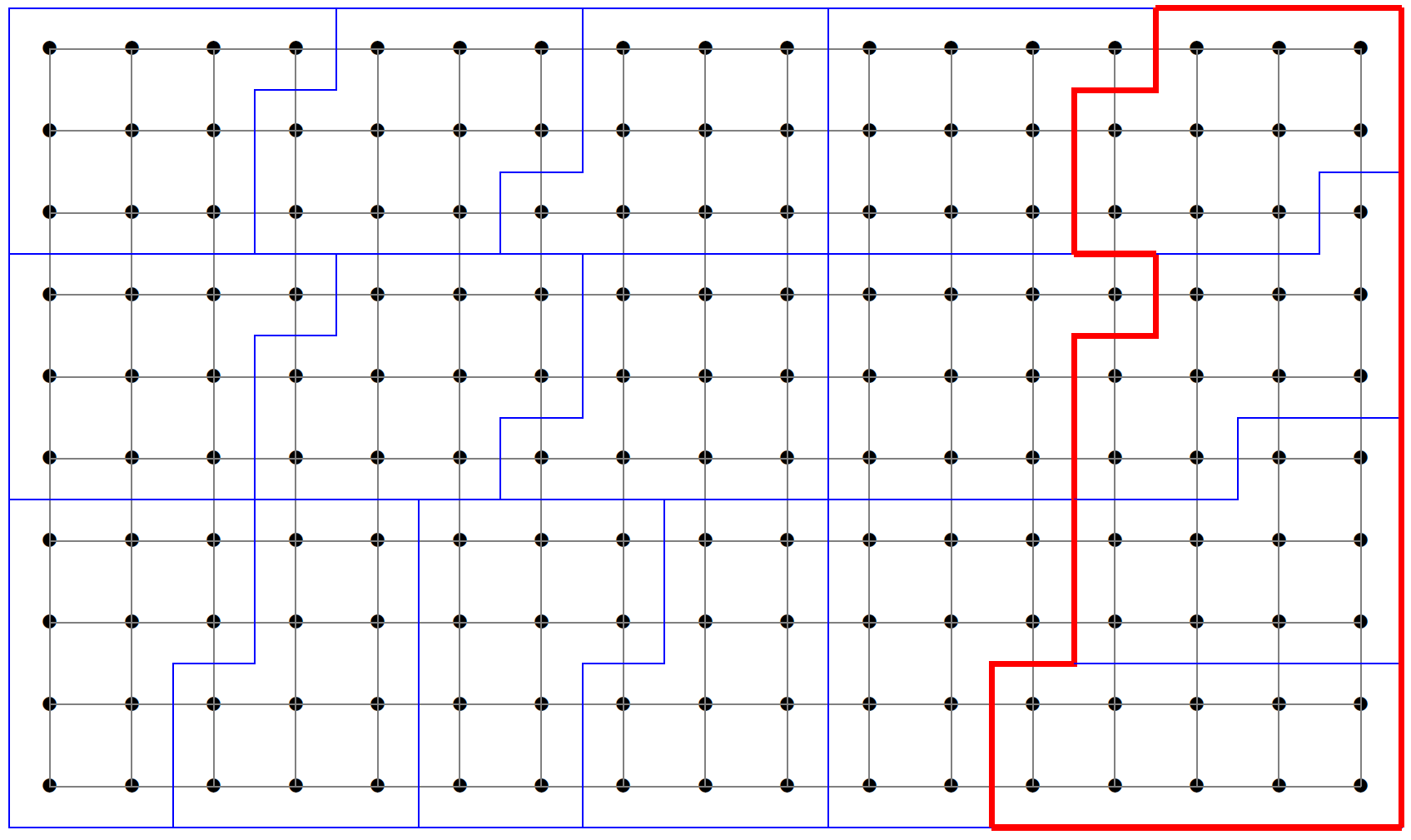}\hspace{0.4cm}
   \includegraphics[width=.48\textwidth]{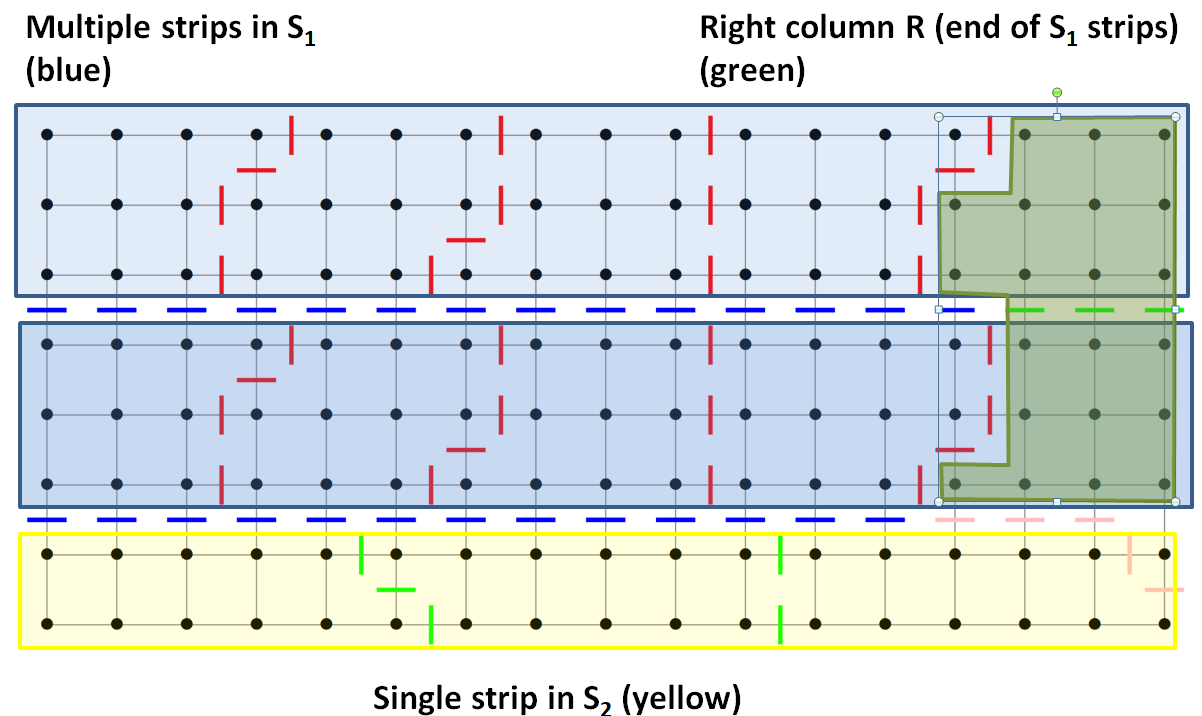}
    \caption{Left: The partition generated by applying {\sc $\varphi$-Cautious Striping} with $m=10, n=17, k=17$. The region $R$ is outlined in red. Right: The regions in the partition generated by applying {\sc $\varphi$-Cautious Striping} with $m=8, n=17, k=13$, with the types of cut edges colored as described in the proof of Theorem~\ref{thm:unweightedstriping}.}
    \label{fig:phistriping}
\end{figure}

We now show a constant-factor multiplicative approximation ratio for the $\phi$-{\sc Cautious Striping} method (Algorithm~\ref{alg:unweighted_striping}).

\begin{theorem}\label{thm:unweightedstriping}
Let $G=(V,E)$ be a rectangular grid graph with $m$ rows and $n$ columns, and $k$ be a desired number of parts, with $n \geq m\geq a$, where $a=\lfloor \sqrt{A} \rfloor$ for $A=\lfloor mn/k \rfloor$. {\sc $\varphi$-Cautious Striping} returns a partition $\mathcal{P}$ of $V$ into contiguous parts of cardinality $A$ and $A+1$. Moreover, if $d,a,\lfloor m/a \rfloor\geq 3$, the number of cut edges of $P$ is no more than $15.25 \OPT$, where $\OPT$ is the minimum number of cut edges over all such partitions.  
\end{theorem}

\begin{proof}
Balance (that each part of $\mathcal{P}$ has cardinality $A$, or $A+1$ if $k\nmid mn$) follows immediately from the construction of the algorithm. 

The striping process and directions are chosen in such a way that contiguity is maintained at every step, with the first vertex of each stripe neighboring a vertex in the rest of the part that the stripe is being added to.
Let $m=da+r$ for some $d\in \mathbb{N},0\leq r\leq a-1$. We first bound $\OPT$ in terms of $m,n,$ and $A$. Recall that the minimum perimeter of a part with $A$ vertices $2\lceil 2\sqrt{A}\rceil$~\cite{christou1996optimal}. Therefore $\OPT_{B}=\frac{mn}{A} \lceil 2\sqrt{A}\rceil$ is a lower bound on the minimum perimeter of all the parts of $\OPT$. Then, subtracting $2(m+n)$ to account for the vertices of degree $2$ and $3$ on the boundary of $G$, we have that $k\geq\frac{mn}{A+1}$ and $\OPT$ is at least 
\[\OPT \geq k \cdot \lceil 2\sqrt{A}\rceil -2(m+n) \geq \frac{A}{A+1} {\rm \OPT}_{B} -2(m+n).\] Furthermore, since $n\geq m\geq da$, the top and bottom borders of $G$, of total length $2n$ are at most $2/(d+1)$ fraction of the perimeter on the top and bottom of parts created by $\varphi$-Cautious Striping. Likewise, the $2m<2n$ cut edges on the left and right borders of $G$ are at most $2/(d+1)$ fraction of the perimeter on the top and bottom of parts created by $\varphi$-Cautious Striping. Thus $2(m+n) \leq \frac{2}{d}\OPT_{B} < \leq \frac{2}{d}\frac{A}{A+1}\OPT_{B}$, and so $\OPT\geq \frac{A}{A+1}\frac{d}{d-2} \OPT_{B}.$ 

We now bound the number of cut edges of the partition given by {\sc $\varphi$-Cautious Striping}. These cut edges come in four types, marked in four colors in Figure~\ref{fig:phistriping}:

\begin{enumerate}
    \item Cut edges between horizontal strips in $S_1$ (blue): There are $d$ rows of cut edges, one below each strip in $S_1$, and each row of cut edges has at most $(n-a/\phi)$ edges in it, for a total of $d(n-a/\phi)$ edges. (Recall that $\phi-1 = 1/\phi$.)
    \item Edges between parts in the same horizontal strip in $S_1$ (red): The number of vertices in $S_1$ is at most $da(n-a/\phi)$, so the number of pairs of adjacent pairs created by the top-to-bottom striping process in $S_1$ is at most $da(n-a/\phi)/A$. Each such pair of parts has at most $(a+1)$ cut edges, which gives a total of at most $da(n-a/\phi)(a+1)/A$ edges. 
    \item Cut edges between parts entirely in $R$, and between parts entirely in $S_2$ (green): The number of vertices in $R\cup S_2$ is at most $(m+n)\phi a$, and therefore the number of parts in $R\cup S_2$ is at most $(m+n)\phi a / A$. Therefore, the number of pairs of such parts with cut edges between them is at most $(m+n)\phi a /A -1$. Each such pair of parts has at most $\phi a +1$ cut edges between them. Therefore, the total number of cut edges of this type is at most
    \begin{align*}
    ((m+n)\phi a /A -1)(\phi a +1) &\leq \phi^2 a^2 (m+n)/A + (m+n)\phi a/A \\
    &\leq \phi^2 (m+n) + \phi(m+n) a/A.
    \end{align*}
    \item Cut edges created after a possible part spanning both $R$ and $S_2$ (pink): There are at most $2\phi a$ such edges.
\end{enumerate}
Summing over each type and rearranging gives the total number of cut edges
%
\begin{align*}
&\quad nd + \phi(a/A+\phi)(m+n) + \frac{a^2 n d}{A} +  2\phi a +\frac{and}{A} - \frac{a^3 d + aAs+a^2 d}{\phi A}\\
&\leq nd + \phi(a/A+\phi)(m+n) +nd+  2\phi a +\frac{and}{A}\\
&\leq 2\frac{mn}{a} + \phi(a/A+\phi)(m+n) +  2\phi a +\frac{mn}{A}\\
&\leq \frac{a+1}{a} \OPT_{B}  + \phi(a/A+\phi)(m+n) +  2\phi a +\frac{mn}{A}. \qquad\qquad(*)
\end{align*}
Dividing the quantity ($*$) by the lower bound on $\OPT$, we have that the approximation ratio is bounded by 15.25 and approaches 1 asymptotically when the number of vertices in the grid and the number of parts required go to infinity with $m/a, n/a \rightarrow \infty$, then we can show the number of cut edges of $P$ approaches $\OPT$. See Lemma \ref{lem:algebra} in the appendix for details.
\end{proof}

Next, we consider a special case of Theorem~\ref{thm:unweightedstriping}, when $k\geq m, n$ and $k|mn$. This is the most general case analyzed by Christou and Meyer~\cite{christou1996optimal}. To achieve a directly comparable result, we state our objective in the same way, using total part perimeter in the dual graph of $G$, rather than total cut edges. This has the effect of counting each interior cut edge twice (since it is in the perimeter of two parts) and adding an additional $2(m+n)$ for the boundary of the grid, which is counted once. Our result, Theorem~\ref{thm:exactunweightedstriping}, shows that $\phi$-{\sc Cautious Striping} yields an asymptotic relative approximation error six times better than that of the $1+9/\lceil 2\sqrt{A}\rceil$ approximation error in~\cite{christou1996optimal} for the case $k>m,n, k|mn$.

\begin{theorem}\label{thm:exactunweightedstriping}
Let $G=(V,E)$ be an unweighted rectangular $m \times n$ grid graph ($m$ rows and $n$ columns, with $n \geq m$ without loss of generality) to be divided into $k$ contiguous parts, where $k|mn$. Then {\sc $\varphi$-Cautious Striping} gives a partition $V_1, \hdots, V_k$ into contiguous parts which achieves exact balance and the following approximation guarantees on compactness: 
Then, if $k\geq \max\{m,n\}\geq 4$, then $\phi$-{\sc Cautious Striping} gives a balanced partition of $V$ into $k$ parts and which has total perimeter at most 
\[
1+\frac{1}{\lceil2 \sqrt{A}\rceil}+\frac{\lceil \phi \sqrt{A}\rceil \cdot\left(\frac{\sqrt{5}}{2}-1\right)+O(1)}{n}
\]
times that of the optimal partition, where $\phi=\frac{1+\sqrt{5}}{2}$. Therefore, the striping algorithm is in the worst case a 1.69-factor approximation. 
\end{theorem}

\noindent 
\textbf{Proof Sketch:} We include the complete proof of the theorem in Appendix~\ref{app:proofs}. Using the assumptions on $m,n,$ and $k$, we show that the $m$ rows can be divided into strips of height $a$ and $a+1$, so in $\phi$-{\sc Cautious Striping}, $S_2 = \emptyset$. We then upper-bound the ratio of the average perimeter of the parts created in each step of the algorithm against \OPT, the lower bound of the perimeter of the parts. Lemma~\ref{lem:uniformstriping} is used to upper-bound the perimeter of the parts in the strips. Averaging this perimeter ratio over all parts of the graph (those created by vertical striping and those in $R$) gives the approximation ratio. \qed

Given a geographical region, one can tile it with square cells (rectangular grid), or hexagonal cells (as in Figure 4). Keeping the dual graph constant for both ways of partitioning the plane, we can use the striping approach for bounding the cut in hexagonal partitions as well. The set of dual vertices is the same, while the size of the cuts changes by adding more adjacent edges. 

We show that the approximation ratio of Theorem~\ref{thm:exactunweightedstriping} can be extended to hexagonal grid graphs. Wang showed that the minimum number of cut edges of an $A$-polyhex (a contiguous group of $A$ hexes) is $2\lceil \sqrt{12A-3}\rceil$~\cite{wangpolyhex} (recall that the minimum number of cut edges of an $A$-polyomino in a square grid graph is $2\lceil 2 \sqrt{A}\rceil$.)
Let $\mathcal{P}$ be a partition of the set of dual vertices such that each part of $\mathcal{P}$ has exactly $A$ vertices.
Let $\textsc{Cut}_S,\textsc{Cut}_H$ be the number of cut edges of $\mathcal{P}$ in the square and hexagonal grid respectively.  If $P$ is a partition generated by $\phi$-{\sc Cautious Striping}, then if two parts in $\mathcal{P}$ have $c$ cut edges between them when $G$ is a square grid graph, those two parts will have at most $2c+1$ cut edges between them in the hexagonal grid.\footnote{For vertices added to their part in a top-to-bottom stripe, vertices with 1 cut edge in the square graph will have at most 2 in the hexagonal graph, and vertices with 2 in the square graph will have at most 3 in the hexagonal graph. For vertices added to their part in a left-to-right stripe, vertices with 1 cut edge in the square graph will alternate between having 1 and 3 cut edges in the hexagonal graph, with those in the middle of stripe that spans both parts having 2.} Let $\textsc{Min}_S = 2\lceil 2 \sqrt{A}\rceil k$ and $\textsc{Min}_H = 2\lceil \sqrt{12A-3}\rceil k$ be lower bounds on the optimal perimeter of a partition into $k$ balanced parts on a square or hexagonal grid, respectively.  If $P$ is a partition generated by $\phi$-{\sc Cautious Striping} on the square grid, then as $A\rightarrow\infty$, we have that
\[
\frac{\textsc{Cut}_H}{\textsc{Min}_H}= \frac{\textsc{Cut}_H}{\textsc{Cut}_S} \cdot \frac{\textsc{Cut}_S}{\textsc{Min}_S}\cdot\frac{\textsc{Min}_S}{\textsc{Min}_H} \rightarrow 2\cdot \alpha\cdot \frac{\sqrt{16}}{\sqrt{48}}= \frac{2}{\sqrt{3}}\alpha
\]
where $\alpha$ is the approximation ratio of $\phi$-{\sc Cautious Striping} for the rectangular grid.

\subsection{Dynamic Striping for General Weights}\label{subsubsec:dynamic} 
The striping techniques in Section~\ref{subsec:uniformstriping} assume constant vertex weights in a rectangular grid graph, as they rely heavily on the geometry of the polyominoes of size $A$. When vertex weights are introduced, the number of vertices within the parts of a good partition can vary greatly, with the total number of cut edges of the optimal partition potentially being incomparable to the unweighted case. Nevertheless, the ideas from striping are useful in the weighted case as well. In this section, we give methods for finding contiguous balanced partitions using dynamic programming techniques, assuming that we have access to a Hamiltonian path $(v_1, \hdots, v_n)$ in the given graph $G = (V,E,w)$ where $w: V \rightarrow \mathbb{R}$ are vertex weights.\footnote{We can further relax this property to require a permutation $(v_1,\dots,v_n)$ (not necessarily a path) of $V$ such that for any interval $[i,j]\subset \mathbb{N}$ with $\sum_{\ell=i}^{j} w(v_i) \geq (1-\varepsilon)A$, the subgraph of $G$ induced by $\{v_i,\dots,v_j\}$ is connected. However, we first present the ideas assuming a Hamiltonian path, for the sake of clarity in exposition.} Note that this path can be obtained by striping for rectangular grid graphs. 

We call a partition $\mathcal{P}=\{V_1,\dots,V_k\}$ of the vertex set $V$ \textit{consistent} with a given ordering $(v_1,\dots,v_n)$, if for all $1\leq i<j\leq k$ and all $u\in V_i$, $v \in V_j$, $u$ precedes $v$ in the ordering $(v_1,\dots,v_n)$. We will search for a balanced, compact partition in a graph which is also consistent with the given Hamiltonian path. Note that there are at most $\binom{n-1}{k-1}$ partitions consistent with a given ordering of vertices as a Hamiltonian path\footnote{To see this, note that given a path, any partition consistent with the path is uniquely determined by the elements of its second through last parts that appear first in the path, and there are $\binom{n-1}{k-1}$ ways to choose these elements, as they can be any of $v_2,\dots,v_n$.}, and if $k$ is a constant then it is easy to find the optimal consistent partition satisfying the approximate balance condition in equation~\eqref{eq:balance}, in polynomial time. A natural question is if we can find the optimal consistent partition for large values of $k$, for e.g., $k = \Theta(n)$. 

To do so, we will use a dynamic programming (DP) approach. We consider the partitions of the prefixes $(v_1,\dots,v_s)$ of the ordered list into a smaller number $t\leq k$ of parts. At each step, we wish to find the consistent partition for the relevant $s$ and $t$ that has minimum cut size. This is non-trivial as the cut function is submodular, and therefore, straightforward DP recursions do not work. The key idea of our approach is to consider cut values on the subgraph $G_s$ of $G$ {\it induced by vertices} $\{v_1,\dots, v_s\}$. Let $j<s$ and let $\mathcal{P}=\{B_1,\dots,B_k\}$ be a partition on $\{v_1,\dots,v_s\}$ with final part $B_k=\{v_j,\dots,v_s\}$. Then the number of cut edges of $\mathcal{P}$ is the number of cut edges between $B_1,\dots,B_{k-1}$, plus the number of cut edges between $B_k$ and the other parts. Using subgraphs, we are able to account for the total in an amortized fashion which helps prove running time guarantees for a natural dynamic program. 

The pseudocode for this algorithm is given in Algorithm~\ref{alg:dynamic_striping}. The algorithm maintains arrays Parts$(s,t)$ and Cut$(s,t)$, both indexed by integers $s$ and $t$, where $0\leq s\leq n$ and $0\leq t\leq k$. Each entry of Parts$(s,t)$ will record the minimum-cut consistent partition of $\{v_1,\dots,v_s\}$ into $t$ parts, and the corresponding entry Cut$(s,t)$ maintains the number of cut edges of Parts$(s,t)$ in $G_s$. To update for a given $s$ and $t$, it considers all $j\leq s$ for which the part $\{v_j,\dots,v_s\}$ satisfies the balance criterion and sets the last part to be $\{v_r,\dots,v_s\}$, where \[r = \argmin_{j: |\sum_{i=j}^{s}w(v_i)-A| \leq \varepsilon A} \text{Cut}(j-1,t-1) + \delta(\{v_1,\dots,v_{j-1}\},\{v_j,\dots,v_s\}).\]

The algorithm runs in polynomial time, since each update takes at most $O(n \mathcal{C})$ time, where $\mathcal{C}$ is the time to compute the cut edges in a given subgraph of $G$. However, we can amortize the cost of computing these successive cuts over each update to show a better run time complexity. 

\begin{algorithm}[!t]\footnotesize
 \caption{(\textsc{Dynamic Partition}) Dynamic partitioning given vertex ordering}
 \begin{algorithmic}[1]
\State \textbf{Input:} Graph $G=(V,E)$ with ordered vertex set $v_1,\dots,v_n$, and vertex weights $w(v)$; desired number of parts $k\in \mathbb{N}$; balance parameter $\varepsilon>0$. Let $G_s$ be the subgraph induced by $v_1, \dots, v_s$ for each $s \leq n$. 
\State \textbf{Output:} A partition of the vertex set into $k$ $(1\pm \varepsilon)$-balanced parts consistent with the ordering $v_1,\dots,v_n$.  
\State \textbf{Initialize:} Parts$(s,t) = \emptyset$, Cut$(s,t) = \infty$, $\forall~ 0\leq s\leq n, 0\leq t\leq k$ except Cut$(0,0)=0$.

\For{$t$ from 1 to $k$}
    \For{$s$ from $t$ to $n-k+t$}
    \State Set $r = \displaystyle\argmin_{j: |\sum_{i=j}^{s}w(v_i)-A| \leq \varepsilon A} \text{Cut}(j-1,t-1) + \delta(\{v_1,\dots,v_{j-1}\},\{v_j,\dots,v_s\}).$
\State Set Parts$(s, t) = $ Parts$(r-1,t-1) \cup \{\{v_i:r\leq i \leq s\}\}$
\State Set Cut$(s,t)$ to the number of cut edges of Parts$(s,t)$ in $G_s$.
    \EndFor
\EndFor
\State \Return Parts$(n,k)$
\end{algorithmic}
\label{alg:dynamic_striping}
\end{algorithm}

Clearly this algorithm runs in polynomial time. In fact, for most applications it runs in $O(|V|^2)$ time. Note that if $G$ is obtained from the regions of a map using the dual graph process of Section~\ref{subsec:formulation}, $G$ is planar and hence satisfies the hypothesis that each vertex has degree $O(1)$. Furthermore, {\sc Dynamic Partition} returns the optimal (min-cut) partition that is consistent with the chosen ordering while satisfying the balance constraint.

\begin{theorem}\label{thm:dynamic} Let $G=(V,E)$ be a graph with ordered vertex set $\{v_1,\dots,v_n\}$, and vertex weights $w(v)$ for all $v\in V$.  Then for desired number of parts $k\in \mathbb{N}$ and balance parameter $\varepsilon>0$, {\sc Dynamic Partition} runs in time $O(|V||E|k)$. Moreover, if $w(v)>0$ and $\mathrm{deg}(v)=O(1)$ for all $v\in V$ and ${\max_v\{w(v)\}}/{\min_v\{w(v)\}}=O(1)$, {\sc Dynamic Partition} runs in time $O(|V|^2)$. Further, {\sc Dynamic Partition} returns a consistent partition $\mathcal{P}$ of $V$ into $k$ parts such that each part has total weight within $\varepsilon$ fraction of $A$, if such a partition exists. $\mathcal{P}$ has minimum cut size over all such consistent partitions, and if $(v_1,\dots,v_n)$ is a Hamiltonian path, each part of $\mathcal{P}$ is connected.
\end{theorem}

\noindent 
\textbf{Proof of running time:} There are $|V|k$ entries of the array $A$ to update. Each update step loops over up to $|V|$ possible vertices as candidates for $j$. For each such candidate $v_j$, we must check the balance condition $(1-\varepsilon)A \leq \sum_{i=j}^{s}w(v_i) \leq(1+\varepsilon)A$ for the part $(v_j,\dots,v_s)$, which takes constant time. We must also update the cut size at each feasible candidate, which consists of checking whether each of its incident edges has its other vertex in $\{v_1,\dots,v_{j-1}\}$. A single edge may be checked in constant time, and the entire update step will require checking at most $2|E|$ edges, since each edge can be checked at most twice (if both its vertices are feasible candidates.) This gives total running time $O(|V|k(|V|+2|E|))=O(|V||E|k)$. 
When computing the $\arg\min$ expression for $r$, beginning at $j=s$ and decrementing, we can compute the running sum $\sum_{i=j}^{s}w(v_i)$ and stop when it exceeds $(1+\varepsilon)A$. If ${\max_v\{w(v)\}}/{\min_v\{w(v)\}}=O(1)$, then the maximum possible number of candidates for $r$ considered by the algorithm in each update step is at most ${(1+\varepsilon)\sum_{v\in V}w(v)}/{\min_{v\in V} w(v)}$ is asymptotically $O(|V|/k)$. Therefore the total running time is $O(|V|k(|V|/k + \Delta(G)|V|/k)) = O(|V|^2)$, where $\Delta(G)$ is the maximum degree of a vertex in $G$.\\

\noindent 
\textbf{Proof of optimality:}  All individual parts created during the execution of {\sc Dynamic Partition} are of the form $\{v_i:r\leq i \leq s\}$ for some $r$ and $s$ such that $(1-\varepsilon)A \leq \sum_{i=r}^{s}w(v_i) \leq(1+\varepsilon)A$. Therefore, the balance condition is satisfied. Moreover, if $(v_1,\dots,v_n)$ is a Hamiltonian path, $\{v_i:r\leq i \leq s\}$ is a path, and therefore connected, for all $r$ and $s$.

We prove optimality over balanced consistent partitions by induction on $t$. In particular, we show that for all $s\in[n]$ and $t \in [k]$, when {\sc Dynamic Partition} terminates,  Parts$(n,k)$ is the minimum-cut and $\varepsilon$-balanced consistent partition of $\{v_1,\dots,v_s\}$ into $t$ parts for all $s$ and $t$. 

For $t=0$, Parts$(s,t)=\emptyset$, which is vacuously optimal for $s=0$ (for $s>0$, no feasible partition exists). Let $t>0$ and suppose the result holds for $t-1$. Let $t\leq s\leq n-k+t$ and let $\OPT=\{\OPT_1,\dots,\OPT_t\}$ be the minimum-cut $\varepsilon$-balanced consistent partition of $\{v_1,\dots,v_s\}$ into $t$ parts (if one exists.) Let $\OPT_t=\{v_{r_{\OPT}},\dots,v_{s}\}$ be the part of $\OPT$ containing $v_s$, and let $\{v_r,\dots,v_s\}$ be the part of Parts$(s,t)$ containing $v_s$. Then we have that
\begin{align*}
\text{Cut}(s,t) &= \text{Cut}(r-1,t-1) + \delta(\{v_1,\dots,v_{r-1}\},\{v_r,\dots,v_s\}) \\
&\leq  \text{Cut}(r_{\OPT}-1,t-1) + \delta(\{v_1,\dots,v_{r_{\OPT}-1}\},\{v_{r_{\OPT}},\dots,v_s\}) \\
&\leq  \delta(\OPT_1,\dots,\OPT_{t-1}) + \delta(\{v_1,\dots,v_{r_{\OPT}-1}\},\{v_{r_{\OPT}},\dots,v_s\}) \\
&=\delta(\OPT). 
\end{align*}

Here the first inequality follows by the definition of $r$ and the $\varepsilon$-balance assumption on $\OPT$ and the second inequality follows by the induction hypothesis. This completes the induction and therefore, the claim. \qed

Clearly the quality of the result of {\sc Dynamic Partition}, particularly the number of cut edges, depends a great deal on the vertex ordering chosen for the input. For grid graphs, one method of obtaining a vertex ordering $(v_1,\dots,v_n)$ derives from the striping technique. Choose a stripe height of $s$ and denote the upper left vertex by coordinates $(1,1)$. In both rectangular and hexagonal grids, denote the $i$th vertex from the top in column $j$ by $(i,j)$. Then the order is $(1,1),(1,2),\dots,(1,s),(2,1),\dots,(2,s),(3,1),\dots.$ 

\begin{figure}[!t]
\centering
  \includegraphics[width=.4\linewidth]{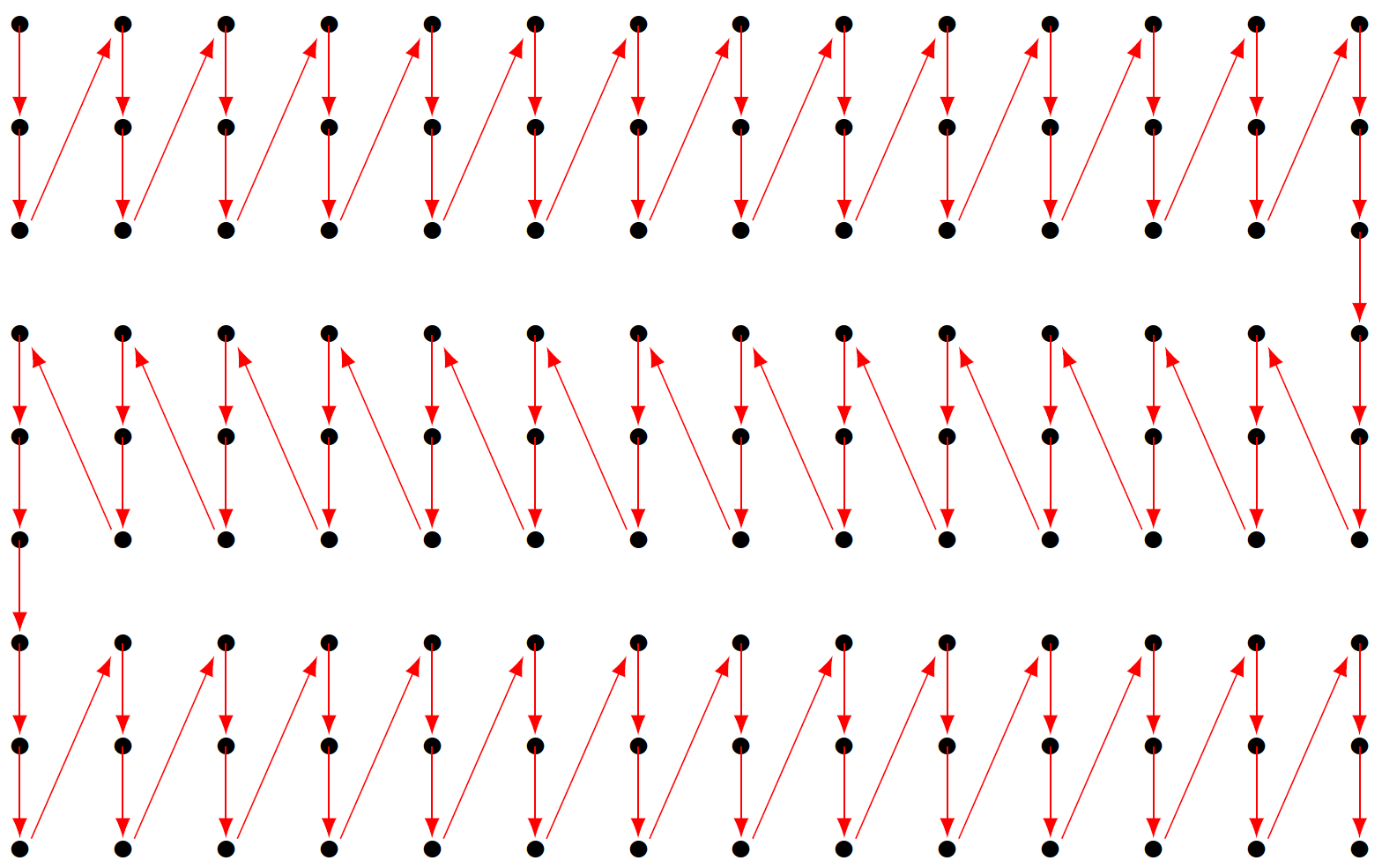}
  \includegraphics[width=.54\linewidth]{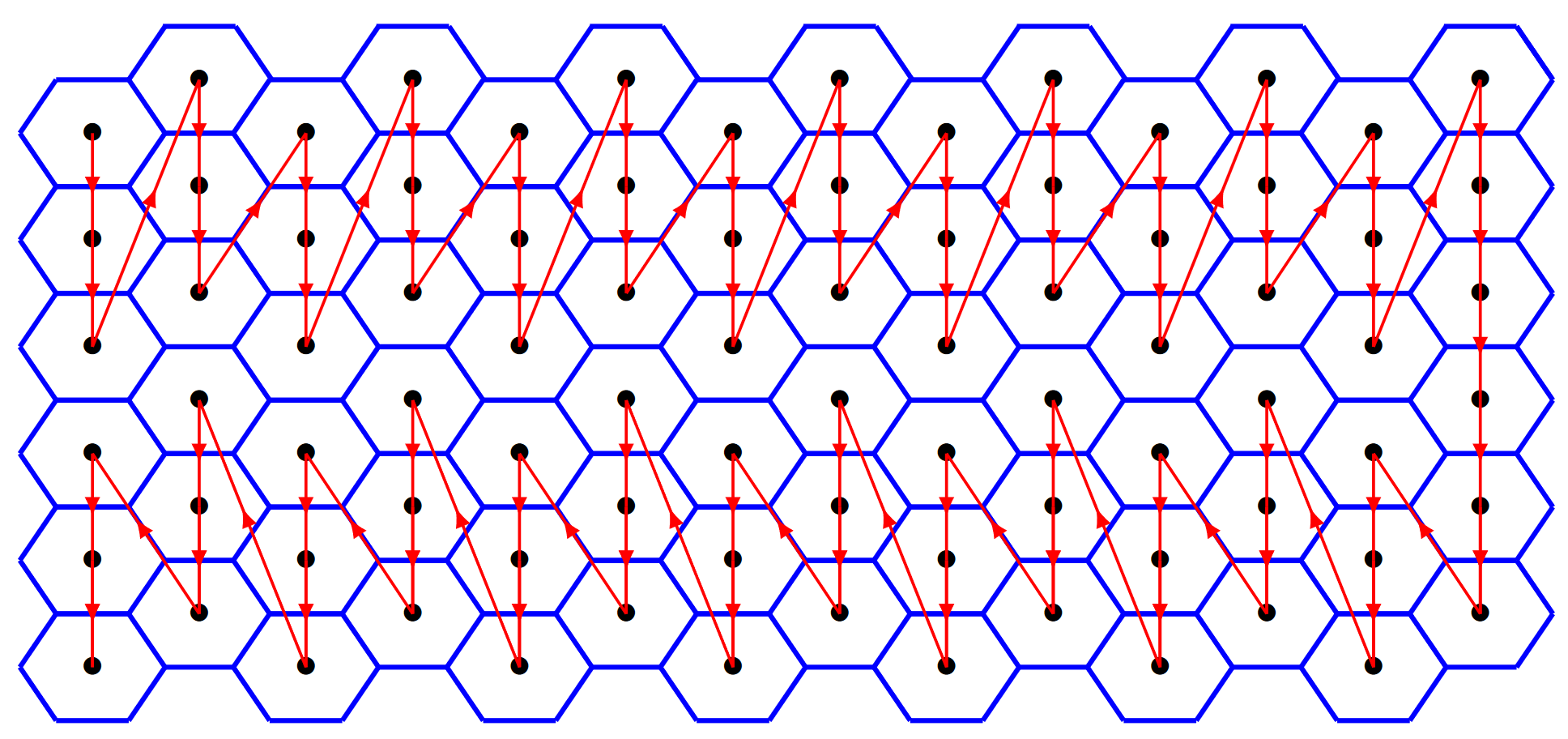}
    \caption{Vertex orders on square and hexagonal grids created using striping and the snake turning technique with stripes of height $3$.}
    \label{fig:snakestriping}
\end{figure}

To maintain connectivity at the end of each strip, we use the snaking idea described in~\cite{christou1996optimal} and illustrated in Figure~\ref{fig:snakestriping}. After traversing the last column of the first stripe, we continue down the column for the height of the next stripe, then proceed by striping right to left. After completing that stripe, we again continue down the last column of the stripe and return to going left to right across the following stripe. An example using this method as input for {\sc Dynamic Partition} is illustrated in Figure~\ref{fig:dynamic_example_weights}.

\subsection{Extensions for Stochastic Weights}\label{subsec:stochastic}
In the stochastic setting, as described in Section~\ref{sec:intro}, exact weights are not known when a graph partition must be found. Instead of showing guarantees on the weight of each part, we wish to show guarantees on the variance of the part weights or the expected maximum part weight, while retaining compactness and contiguity guarantees.



In particular, this can apply if the $X_v$ are taken to be Poisson distributions, which can model the occurrence of individual events, such as 911 calls, that contribute to a vertex's workload in applications. When each vertex weight is drawn from an independent distribution, these random variables $X_v$ have variance  $\sigma^2$ bounded by a constant factor $c\mu(v)$ of their mean. In this case, if we have a partition of $V$ such that the expected weight of each part is close to the average expected weight $A$, the variance of the weights of the parts is upper-bounded as well.

\begin{theorem}[Stochastic weights: Implications.]\label{thm:variance_bound}
Let $G=(V,E)$ be a graph. For each $v\in V$, let $X_v$ be an (independent) random distribution with mean $\mu(v)$ and variance $\sigma^2 \leq c\mu(v)$ for some given constant $c$. Suppose $V_1,\dots, V_k$ is a partition of $V$ such that for some $\varepsilon\geq 0$, $\left|\sum_{v\in V_i} \mu(v)-A\right|\leq \varepsilon A$ for all $i$, where $A(\mu)=(1/k)\sum_{v\in V}\mu(v)$. Then 
\[
\frac{1}{k}\E\left[\sum_{i=1}^k \left(\sum_{v\in V_i} X_v - A\right)^2\right] \leq (4\varepsilon + \varepsilon^2)  A^2 + cA.
\] 
\end{theorem}
\noindent 
\textbf{Proof Sketch:} The proof is given in Appendix \ref{app:proofs}. The key idea is to bound the variance $\frac{1}{k}\E\left[\sum_{i=1}^k \left(\sum_{v\in V_i} X_v - A\right)^2\right]$ as the sum of two terms. One results from the variances of the variables $X_v$, and is bounded by $ckA$ since the $X_v$ are independent. The other term results from the discrepancies in the expected total weights of the parts (up to $\varepsilon$ each) and is bounded by $(4\varepsilon + \varepsilon^2) k A^2$. \qedsymbol

To adapt partitioning methods to the stochastic model, we derive vertex weights from the distributions $X_v$, incorporating a proxy for the variance of the distributions. Our goal is to bound the expected maximum part weight $\E[\max_i \sum_{v\in V_i} X_v]$. To do this, we adapt the methods of Kleinberg et al.\ for the stochastic load-balancing problem, which apply for general stochastic weights~\cite{kleinbergbursty}. This algorithm iteratively divides each distribution $X_v$ into a ``normal part" $N_v = X_{v}\cdot \mathds{1}_{\{X_v\leq 2^i\}}$ and an ``exceptional part"  $S_v=X_v \cdot \mathds{1}_{\{X_v > 2^i\}}$. At each iteration, if the total expected value of the exceptional parts is at most 1, apply \textsc{Dynamic Partition} for a balanced, contiguous partitioning.\footnote{In this setting, when solving the recursion in \textsc{Dynamic Partition} we enforce only the upper constraint on total part weight, not the $(1-\varepsilon)A$ lower constraint. Note that the running time of \textsc{Dynamic Partition} is still polynomial.} We use the transformation transformation $\beta_{1/k}(N_v)= \log \E[k^{N_v}]/\log k$ to obtain the vertex weights for all $v\in V$. This yields a contiguous partition, with $\E[\max_i \sum_{v\in V_i} X_v] = O(A)$. A full description is given in Algorithm~\ref{alg:stochastic_load_balancing}.

\begin{algorithm}[t!]\footnotesize
\caption{(\textsc{Stochastic}) {Stochastic Load Balanced Partitioning}}
 \begin{algorithmic}[1]
\State \textbf{Input:} Graph $G=(V,E)$, with independent random weights $X_v$ ($X_v\geq 0$, $v\in V$), number of parts $k$. 

\State \textbf{Output:} Partition $\{V_1, \hdots, V_k\}$ s.t. $\E[ \max_{1\leq i\leq k} \sum_{v\in V_i} X_v]$ is an $O(1)$ approximation.


\State Define $N_{v,i} = X_{v}\cdot \mathds{1}_{\{X_v\leq 2^i\}}$, $S_{v,i} = X_v \cdot \mathds{1}_{\{X_v > 2^i\}}$ for $i \in \mathbb{N}$. 

\State Let $w_{v,i} \leftarrow \beta_{1/k}(N_{v,i}/2^i)= \log \E[k^{N_{v,i}/2^i}]/\log k$, $\varepsilon_i = (18k-\sum_{v}w_{v,i})/\sum_{v}w_{v,i}$.

\State $i^* = \min \{i: \sum_{v\in V} \E[S_{v,i}/2^i] \leq 1$ and {\sc Dynamic Partition}$ (G, k, w_{v,i} (v\in V) ,\varepsilon_{i})$ is feasible.\}

\State Return {\sc Dynamic Partition}$ (G, k, w_{v,i^*} (v\in V) ,\varepsilon_{i^*})$. 
\end{algorithmic}\label{alg:stochastic_load_balancing}
\end{algorithm}

In general, the partition produced by {\sc Stochastic Partition} will inherit any graph-theoretic guarantees on partitions produced by {\sc Dynamic Partition}. If $G$ has a Hamiltonian path, using that path as input to {\sc Dynamic Partition} will produce a contiguous partition. While most graphs produced from maps, including rectangular grid graphs, will have Hamiltonian paths, in general such an guarantee may not exist if $G$ does not have a Hamiltonian path. \footnote{For instance, if $G$ is the spider graph on 28 vertices composed of three paths of length 10 which share one end vertex, then $G$ has no Hamiltonian path~\cite{Diestel}.} Indeed, there is no partition of $G$ into two contiguous parts, each with no more than 18 vertices. We show that {\sc Stochastic Partition} is an $O(1)$-approximation algorithm for the expected maximum part size of the partition. We include a discussion of why stochastic load balancing results carry over to the stochastic partition \cite{kleinbergbursty}, in Appendix \ref{app:proofs}.  

\begin{theorem}[Stochastic weights: General]\label{thm:load_balancing}
Let $G=(V,E)$ be a graph with an independent nonnegative random variable $X_v$ for each vertex $v\in V$, a desired number of parts $k$. Then {\sc Stochastic Partition} returns a partition $\mathcal{P}$ such that $\E[ \max_{1\leq i\leq k} \sum_{v\in V_i} X_v]= O(1)\cdot \OPT_S,$ where $\OPT_S$ is the minimum of $\E[ \max_{1\leq i\leq k} \sum_{v\in B_i} X_v]$ over all partitions of $V$ into $k$ parts $B_1,\dots,B_k$.
\end{theorem}

\noindent 

\section{Computational Results}\label{sec:computations}

In this section, we examine the performance of our dynamic and stochastic graph partitioning algorithms. We test synthetic settings and real-world instances for the fire and police departments of South Fulton, Georgia. For the instance sizes we consider in this work, traditional approaches using mixed-integer formulations cannot scale even with striping based warm-starts. For example, we found that the full-size South Fulton districting instances were computationally intractable, with little improvement found by the IP even after 20 hours of run time. We simplified the South Fulton police model by increasing hexagonal grid diameter to approximately 1 mile, resulting in a graph with 169 vertices. Running {\sc Dynamic Partition} on this instance with $k=7$ parts and balance $\varepsilon = 0.1$ produced a partition with 191 cut edges. Using this partition as an initialization, after 15 hours the IP had a gap of 40.1\% between its lower bound and best incumbent result (a partition with 188 cut edges.) Instead, we consider the striping approaches in combination with simulated annealing to get tractable experimental results. In addition, we compare the results to benchmarks obtained using $k$-means partitioning, given the planarity of our instances.

\subsection{Comparative Methods: Simulated Annealing}
\label{sec:annealing} 



Mathematical programming models \cite{sauglam2006mixed, ding2015mixed, yang2015milp} are essential tools for modeling and solving combinatorial optimization problems in numerous fields, which can guarantee the optimality of the obtained solutions, are mostly based on mixed-integer linear programming (MILP). However, when the problem involves a large number of variables, MILP can become very expensive to solve. A metaheuristic method, simulated annealing (see, e.g., \cite{bertsimas1993simulated}), has been widely adopted in solving the combinatorial optimization problem. The simulated annealing algorithm explores the neighborhood of the current solution and decides a better substitution randomly. Simulated annealing can achieve reasonable performance in practice for various settings, although there are very limited theoretical performance guarantees \cite{aarts1987simulated, van1987simulated,lecchini2008simulated}. 
In particular, in our setting, we use the current/existing partition as an initial solution. Based on this, a new solution can be founded by selecting from a set of candidate solutions. The set of candidate solutions is typically constructed as ``neighboring'' solutions to the current solution without breaking contiguity.

Specifically, in the $n$-th iteration, our simulated annealing algorithm performs the following acceptance-rejection sampling. Suppose the starting partition is $\mathcal{P}_n$. For instance, we can take the existing partition as an initialization. The next partition $\mathcal{P}_{n+1}$ is selected from a set of candidate partitions defined as $\mathcal S_{n+1}$ and $\mathcal{P}_{n+1} \in \mathcal{S}_{n+1}$.  
The candidate partitions in $\mathcal S_{n+1}$ satisfy contiguity and balance constraints. We randomly choose one of these candidate partitions $\mathcal{P}_{n+1} \in \mathcal S_{n+1}$, and evaluate a score   %
\[
    P(\mathcal{P}_{n+1}, \mathcal{P}_n|T) = 
    \begin{cases}
    1, & Z(\mathcal{P}_{n+1}) < Z(\mathcal{P}_n),\\
    \exp\{|Z(\mathcal{P}_{n
    +1}) - Z(\mathcal{P}_n)|/T\}, & \text{otherwise},\\
    \end{cases}
\] 
where $Z(\cdot)$ denotes the cost associated with a partition (e.g., the compactness shown in \eqref{eq:compactness}), $T$ is a pre-specified temperature parameter that determines the speed of convergence, and the $P$ is the acceptance probability. We generate an independent uniform random variable $U \in [0, 1]$. 
The proposed partition is accepted if $P(\mathcal{P}_{n+1}, \mathcal{P}_n|T) \ge U$. We refer to an update of the proposed partition as a \emph{transition}. Note that there is a chance that the transition happens from a ``low-cost'' partition to a ``high-cost'' partition, and this ``perturbation'' will prevent the algorithm from being trapped at a local sub-optimal solution. The choice of the set candidate partitions $\mathcal S_{n+1}$ is critical for the performance of simulated annealing, which involves the trade-off between exploration and exploitation. We next discuss two strategies for creating candidate partitions based on square and hexagonal grids.
\begin{figure}[t]
\centering
\begin{subfigure}[h]{0.47\linewidth}
\includegraphics[width=\linewidth]{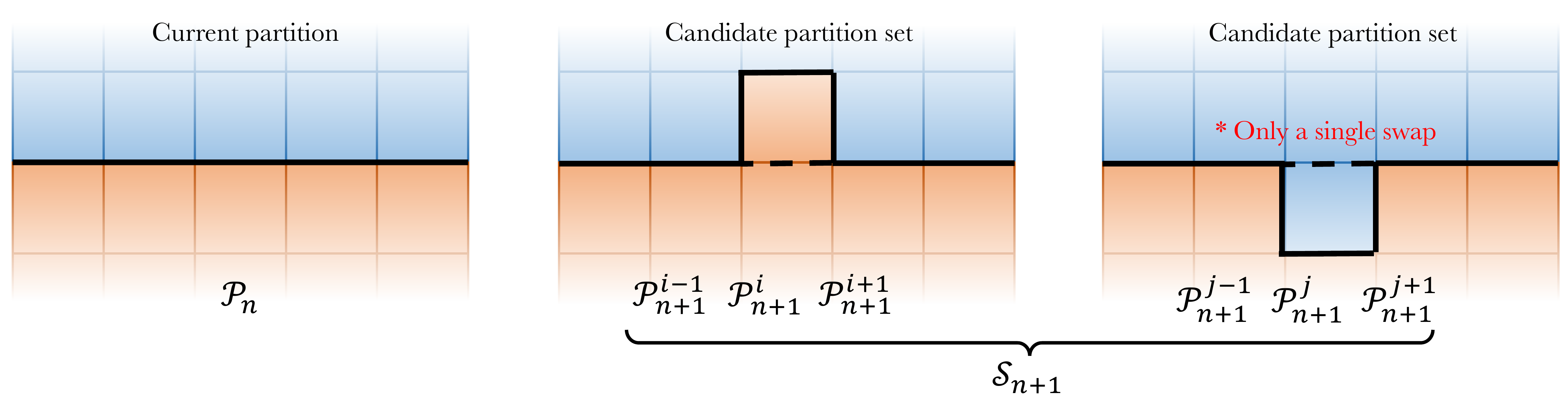}
\caption{One-swapping square neighborhoods}
\end{subfigure}
\hfill
\begin{subfigure}[h]{0.47\linewidth}
\includegraphics[width=\linewidth]{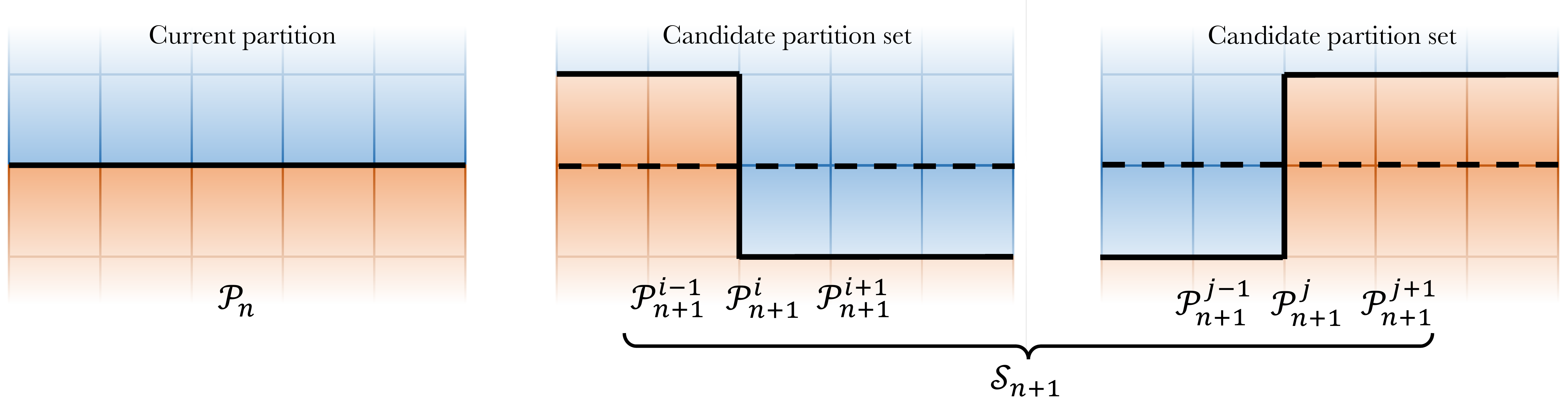}
\caption{Combinatorial square neighborhoods}
\end{subfigure}
\vfill
\vspace{.1in}
\begin{subfigure}[h]{0.47\linewidth}
\includegraphics[width=\linewidth]{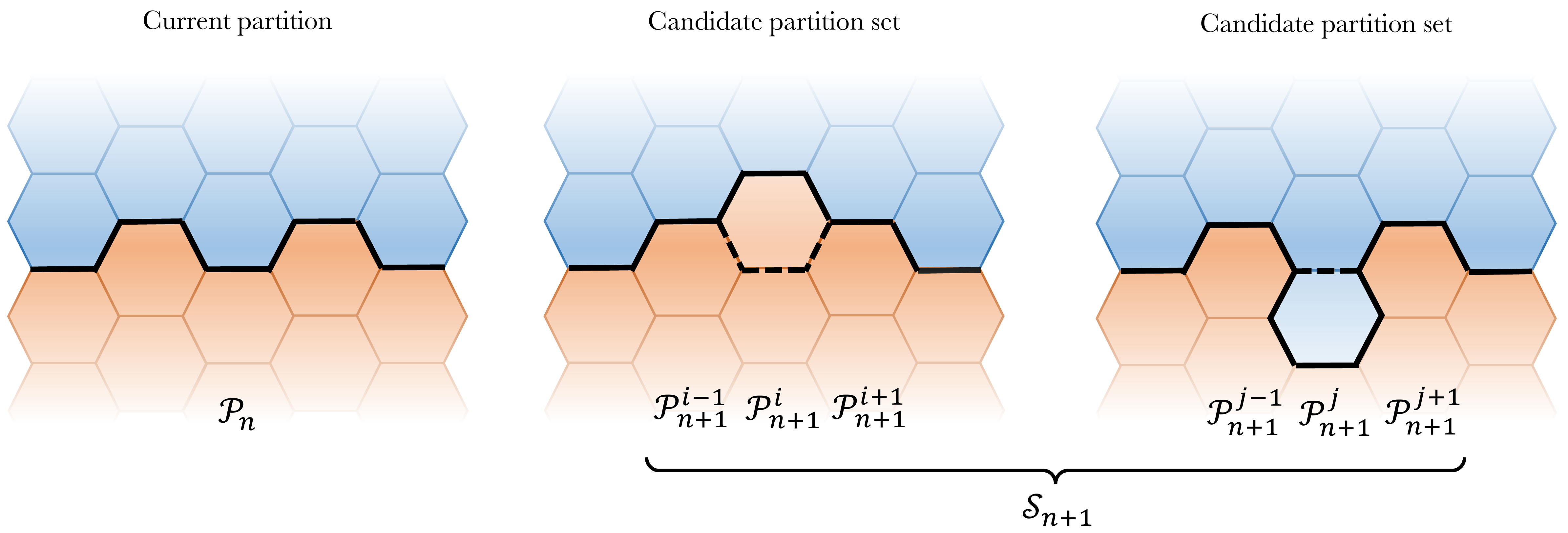}
\caption{One-swapping hexagonal neighborhoods}
\end{subfigure}
\hfill
\begin{subfigure}[h]{0.47\linewidth}
\includegraphics[width=\linewidth]{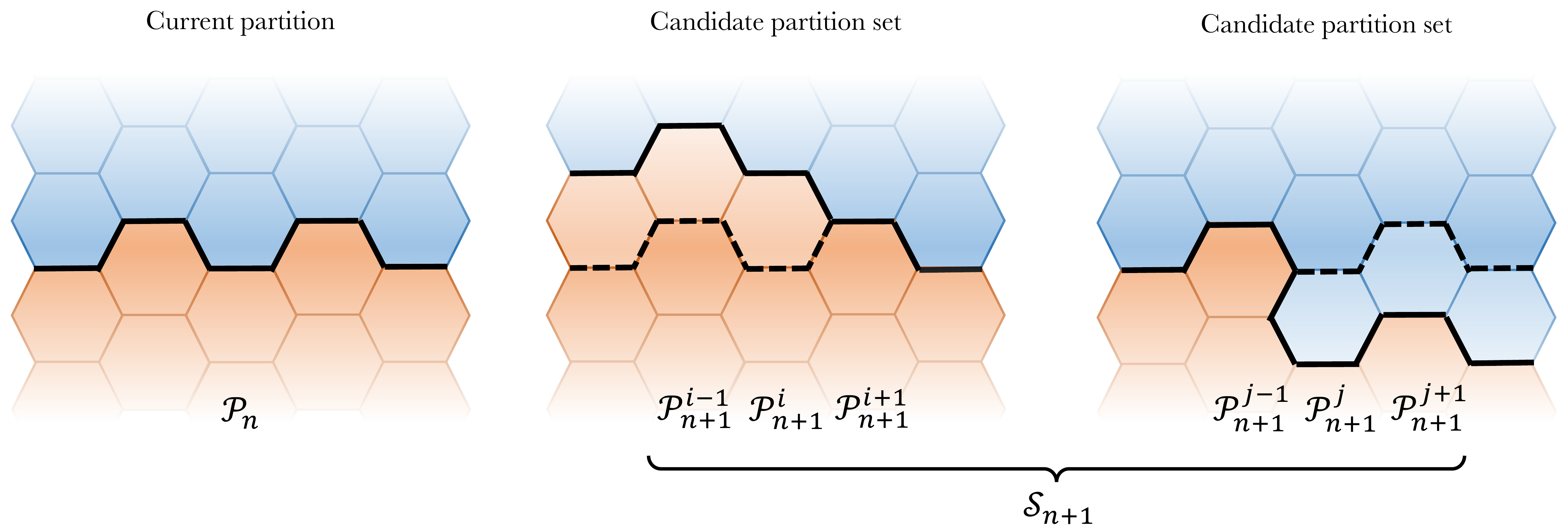}
\caption{Combinatorial hexagonal neighborhoods}
\end{subfigure}
\caption{Illustration of two approaches for candidate partitions based on one-swapping neighborhoods and combinatorial neighborhoods, respectively. Red and blue boxes represent vertices in different parts. The thick black line represents the boundary of two parts. In each subfigure, their left panel shows the current partitions $\mathcal{P}_n$; the middle and the right panels show the candidate partition sets for the next iteration $\mathcal{S}_{n+1}$.}
\label{fig:sa-demo}
\end{figure}

\emph{One-swapping neighborhoods.}
As illustrated in Figure~\ref{fig:sa-demo} (a, c), we first consider the following simple heuristic in constructing the candidate set. This allows us to search for local optimal partitions at a reasonable computational cost~\cite{Wang2013}. The candidate set contains all partitions that swap a single vertex assignment at the boundary of the current partition. This simple heuristic is easy to implement since the number of such candidate partitions is usually small (because we only swap one end of the boundary edge). However, such candidate sets may contain partitions that are still too similar to the current partition. Therefore, we will consider the following alternative strategy.

\emph{Combinatorial neighborhoods.}
Consider the following combinatorial approach to construct a larger candidate set than the previous approach: $\mathcal S_{n+1}$ contains partitions with multiple swappings of vertices assignments of the current configuration, as illustrated in Figure~\ref{fig:sa-demo} (b, d). This combinatorial strategy is more efficient in search of compact partitions than the one-swapping strategy,  as shown in Section~\ref{sec:computations}. 

Moreover, we note that the hexagonal grid results in better performances than the square grid in practice. This can be explained by the fact that each vertex has six neighbors in the hexagonal grid, permitting more ``exploration'' for the simulated annealing algorithm than in the square grid. Furthermore, candidate swaps between two parts with a straight border are more likely to decrease the number of cut edges on a hexagonal grid than on a square grid.

\subsection{Weighted Planar Graphs}
\label{sec:generalcomputations}

In this section, we demonstrate the benefits of the striping method {\sc Dynamic Partition}. Consider a $100\times 100$ hexagonal grid. The distribution of weights on the grid, shown in Figure~\ref{fig:dynamic_example_weights}, are generated as follows. First, all vertices $v\in V$ are given a random weight $w_v$ sampled independently from  $U[0,1]\cdot b_v$, where $U[0,1]$ denotes the uniform distribution on $[0,1]$ and $b_v$ is a Bernoulli random variable with parameter 0.02. Next, we apply the following Gaussian kernel smoothing ~\cite{wasserman2006all} forty times to interpolate the weights on the grid:
\[
w(i,j)\leftarrow \left(\frac{\sum_{k=i-3}^{i+3}\sum_{\ell=j-3}^{j+3}\exp\left(-\sqrt{|i-k|^2+|j-\ell|^2}\right)\cdot w(k,\ell)}{\sum_{k=i-3}^{i+3}\sum_{\ell=j-3}^{j+3}\exp\left(-\sqrt{|i-k|^2+|j-\ell|^2}\right)} \right).
\]

Given these weights, we partition the graph into $k=100$ parts with balance parameter $\varepsilon=.05$, and a Hamiltonian path constructed using the striping method with strip height $s=10$. The {\sc Dynamic Partition} gives a partition depicted in Figure~\ref{fig:dynamic_example_weights}, with each part represented with a different color. Note that the partition is contiguous, balanced, and compact, and the total number of cut edges is 0.5\% greater than the compact but unbalanced partition {where each part is a $10\times10$ square of vertices.} Subsequently, we use simulated annealing with combinatorial neighborhoods search, initialized with the striping partition. The preset temperature is 0.5. The resulting compactness and running time metrics of these partitions are summarized in Table~\ref{tab:computations-artificial}. Simulated annealing improved the number of cut edges from 3601 to 3570 and from 3528 to 3498 with $\varepsilon = 0.02$ and $\varepsilon = 0.05$, respectively. On a single machine with an Intel i7-4770 processor, for $\varepsilon = 0.02$ the running times for {\sc Dynamic Partition} and simulated annealing were 27 minutes and two hours, 25 minutes, respectively; and 40 minutes  and 59 minutes for $\varepsilon = 0.05$, respectively. Simulated annealing does not significantly change the balance, with all partitions having some parts close to the limits $(1-\varepsilon)A$ and $(1+\varepsilon)A$. 
\begin{figure}
    \includegraphics[width=.99\textwidth]{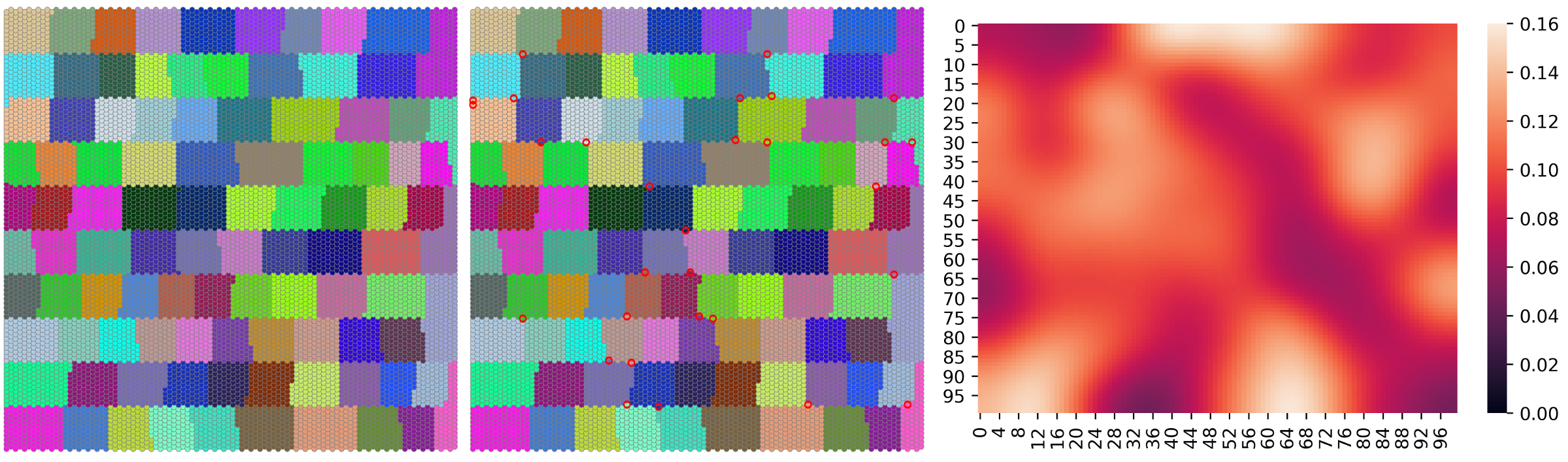} 
    \caption{A contiguous, balanced, compact partition of the $100\times100$ hexagonal grid into 100 parts with $\varepsilon = .05$ using simulated annealing based on combinatorial neighborhoods.}
    \label{fig:dynamic_example_weights}
\end{figure}



\subsection{Stochastic Weighted Grid Graphs}\label{sec:synthetic-stochastic}

The next synthetic example demonstrates the effectiveness of the {\sc Stochastic Partition} algorithm when the vertex weights are drawn from random distributions and how changing the empirical sample size affects performance. On a $60\times 60$ square grid, each vertex $v$ was assigned a generalized weight distribution $G_v = \mathrm{GEV}(\mu,\sigma,\xi)$ with probability distribution function $f_v$, with the location, scale, and shape parameters $\mu,\sigma,\xi$ respectively varying randomly over the grid, as well as an additional random scale parameter $s_v$. For computational efficiency, each $G_v$ was then discretized to a distribution $X_v$, where
\[
\Pr[X_v = i\cdot s_v] = \begin{cases} f_v(i)/\sum_{j=1}^{250} f_v(j), i \in \{1,2,\dots,250\}\\0, \text{otherwise.}\end{cases}
\]

Four types of partitions were computed: true fixed and true stochastic, using {\sc Dynamic Partition} and {\sc Stochastic Partition} respectively with the $X_v$ as input; and empirical fixed and empirical stochastic, using an empirical mean and {\sc Stochastic Partition} (on empirical distributions) respectively. Fourteen samples were taken from each $X_v$, and for $2\leq i \leq 14$, the inputs for the empirical fixed and stochastic partitions used the first $i$ samples from each vertex. Parameters $k=75$ and $\varepsilon=0.1$ were used for all partitions.

To evaluate the performance of the partitions, a sample $w(v)$ was drawn from each true distribution $X_v$ and the normalized maximum part weight
\[
\frac{k \cdot\max_{P_i \in \mathcal{P}} \sum_{v\in P_i} w(v)}{\sum_{v\in V} \E[X_v]}.
\]

 This process was repeated for 100 different sets of $X_v$ and 1000 sets of samples from the true distribution for each such set. Figure~\ref{fig:stochastic_gev} shows the normalized maximum part weight for each sample size and partition type averaged over all $10^6$ corresponding instances, expressed as the percent deviation from the normalized part weight of each instance. We note that {\sc Stochastic Partition} may not always select the number $i^*$ of weight halving steps that produces the best partition. For each set of true distributions and each empirical sample size $2\leq i \leq 14$, {\sc Stochastic Partition} was run with 10 to 13 halving steps. The resulting partition with smallest normalized maximum part weight over all random samples $w(v)$ was then selected, and the average maximum part weight of these partitions is also plotted as the dotted line in Figure~\ref{fig:stochastic_gev}.

 \begin{figure}[!t]
    \centering
        \includegraphics[width=.45\textwidth]{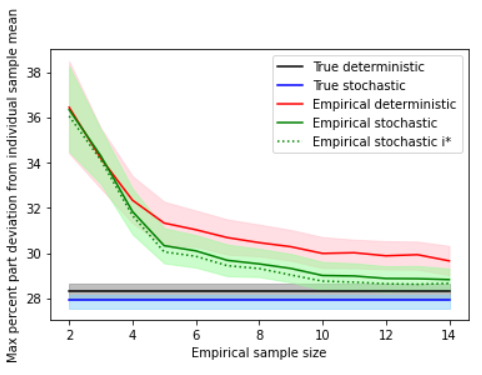}
          \includegraphics[width=0.45\textwidth]{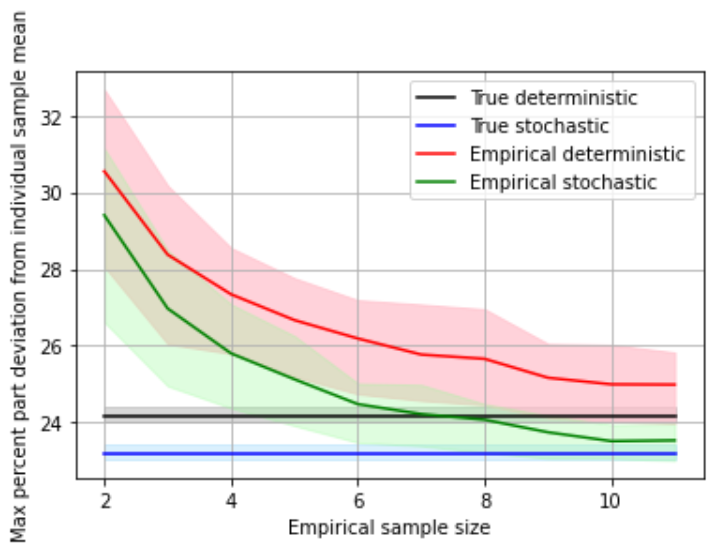}
        \caption{The maximum part weights in partitions created by {\sc Stochastic Partition} (green/black), {\sc Dynamic Partition} (red/blue) with various empirical sample sizes and using the true distributions. (Left) The average maximum part weight in partitions  with each of the 100 sets of distributions on the $60\times 60$ grid, as well as by {\sc Stochastic Partition} using the optimum value of $i^*$. Shaded areas show the 90\% confidence interval of each statistic. (Right) The average maximum part weight in partitions over the 100 sets of distributions on the hexagonal grid over South Fulton,as described in Section~\ref{sec:case_study_police}. Shaded areas show the two middle quartiles of each statistic.}
        \label{fig:stochastic_gev}
\end{figure}
\subsection{Case Study: City of South Fulton}\label{sec:case_study}
The City of South Fulton, Georgia, was established in May 2017 from previously unincorporated land outside Atlanta. It is now the third-largest city in Fulton County, with a population of over 98,000, of which 91.4\% are black or African-American~\cite{UScensusbureau}. In this section, we consider two case studies involving applications of our methods to designing partitions of the city into regions for the South Fulton fire department and police department's operation. The goal in these applications is to obtain balanced workloads in regions while minimizing travel time and improving coverage (in the fire department case study), and reducing over-policing (in the police department case study.) 

These problems have a natural set up for our approach, in which balance, contiguity, and compactness are all necessary for designing a good partition. We discretize the space to cover the city by a hexagonal grid approximately 0.34 miles in diameter for both cases. Each hexagon in the mesh is a vertex in the districting model's grid graph, containing 1236 vertices. Workloads are based on past data from 911 calls and translate to vertex weights in the model. We compare our approach to weighted $k$-means \cite{Kerdprasop2005} as a baseline, where the police workload of each vertex is treated as the weight.

\subsubsection{Fire department workload balancing}\label{sec:case_study_fire}

We include the current station locations and their coverage in Figure~\ref{fig:SF_fire_map} (see Appendix~\ref{append:south-fulton}). Currently, there are ten fire stations in the City of South Fulton. We overlay a hexagonal grid, and each grid point's workload is estimated using the 911 fire call data, as defined as the average time responding to the fire incidents each day. The resulting vertex weights range from an average of 0 to 68.78 minutes per day. {\sc Dynamic Partition} algorithm on the grid graph with $k=15$ and $\varepsilon=0.1$ gives partitions as shown in Figure \ref{fig:fire-data}. We compare our results to partitions created using $k$-means methods to minimize the number of cut edges. We also consider perturbed workloads that are increased by 0.1 and 0.5 on each vertex. This ensures that vertices with no observed calls still have a nonzero associated workload. Table~\ref{tab:computations-fire-data} displays the total number of cut edges and balance for each set of weights. The corresponding partitions are presented in Figure \ref{fig:fire-data-appendix}, Appendix~\ref{append:south-fulton}. These results confirm that our approach can obtain the most compact partitions with balanced workload compared to the baseline method, and SA can further improve the compactness of the generated partitions while maintaining the balance of the workload. In particular, while $k$-means obtained a partition with the minimum cut edges (452), there exist parts with almost twice the amount of workload compared to others (i.e., max deviation of 71.3\%). This was significantly improved by striping (to 9.3\%) using an increased cut of size 690. Further, striping with SA improved the cut by 16\%, while staying within the desired balance.

\begin{figure}[t]
\centering
\begin{subfigure}[h]{0.23\linewidth}
\includegraphics[width=\linewidth]{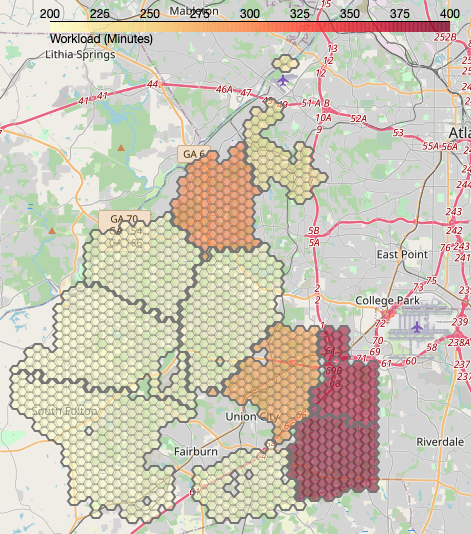}
\caption{$k$-means}
\end{subfigure}
\begin{subfigure}[h]{0.23\linewidth}
\includegraphics[width=\linewidth]{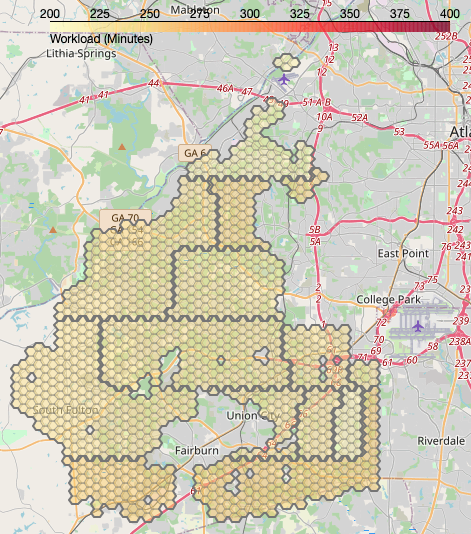}
\caption{Striping}
\end{subfigure}
\begin{subfigure}[h]{0.23\linewidth}
\includegraphics[width=\linewidth]{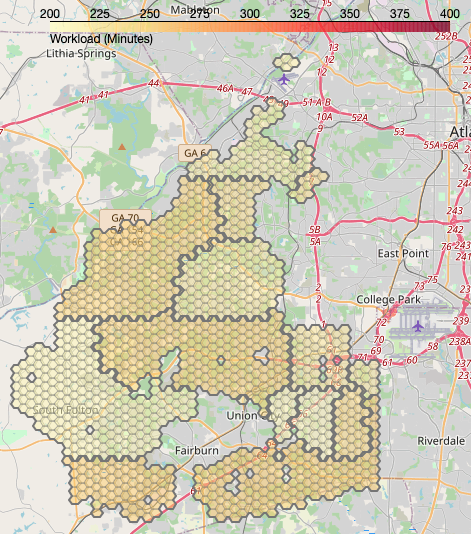}
\caption{Striping+SA}
\end{subfigure}
\caption{Fire station districting results, with deeper colors depicting higher workloads and performance metrics summarized in Table~\ref{tab:computations-fire-data}.}
\label{fig:fire-data}
\centering
\vspace{.15in}
\resizebox{\textwidth}{!}{%
\begin{tabular}{|c|c|c|c|c|c|c|c|c|c|}
\hline
& \begin{tabular}[c]{@{}c@{}}$k$-means\\ (base)\end{tabular} & \begin{tabular}[c]{@{}c@{}}$k$-means\\ (+.1)\end{tabular} & \begin{tabular}[c]{@{}c@{}}$k$-means\\ (+.5)\end{tabular} & \begin{tabular}[c]{@{}c@{}}Striping\\ (base)\end{tabular} & \begin{tabular}[c]{@{}c@{}}Striping\\ (+.1)\end{tabular} & \begin{tabular}[c]{@{}c@{}}Striping\\(+.5)\end{tabular} & \begin{tabular}[c]{@{}c@{}}Striping+SA\\(base)\end{tabular} & \begin{tabular}[c]{@{}c@{}}Striping+SA\\(+.1)\end{tabular} & \begin{tabular}[c]{@{}c@{}}Striping+SA\\(+0.5)\end{tabular}\\ \hline
Total cut edges & {\bf 452} & {\bf 440} & 412 & 690 & 704 & 694 & 596 & 588 & 564 \\ \hline
\begin{tabular}[c]{@{}c@{}}Max deviation\\ from mean\end{tabular} & 71.3\% & 66.9\% & 61.3\% & {\bf 9.3\%} & {\bf 9.2\%} & {\bf 9.8\%} & 9.8\% & 9.9\% &9.9\% \\ \hline
Run time (m:s) & 00:01 & 00:02 & 00:01 & 04:02 & 04:00 & 3:48 & 13:18 & 13:04 & 13:22  \\ \hline
\end{tabular}}
\captionof{table}{Performance k-means, striping, and SA with striping as a warm-start on fire station data. Additive weights of $+0.1, +0.05$ help regularize areas with 0 observations.}
\label{tab:computations-fire-data}
\vspace{-2mm}
\end{figure}

\subsubsection{Police districting and reducing over-policing}\label{sec:case_study_police}

The City of South Fulton observed climbing crime rates and long police response times. Thus, the South Fulton City Council made it clear that their number one priority was to make the city safer \cite{SFStrategicPlan}. This is partly due to the demographic and traffic pattern changes, which create an unbalanced workload among different regions. Figure~\ref{img:South Fulton quantities} in Appendix~\ref{append:south-fulton} shows the distribution of 911 calls, which we estimated from 911-call data provided by SFPD from 2018 to 2019. It is evident from the figure that certain beats faced a significantly higher workload than others.

\begin{figure}[t]
\resizebox{\textwidth}{!}{%
\begin{tabular}{|c|c|c|c|c|c|c|c|c|}
\hline
 & \begin{tabular}[c]{@{}c@{}}$k$-means\\ (7 parts)\end{tabular} & \begin{tabular}[c]{@{}c@{}}$k$-means\\ (15 parts)\end{tabular} & \begin{tabular}[c]{@{}c@{}}Striping\\ (7 parts)\end{tabular} & \begin{tabular}[c]{@{}c@{}}Striping\\ (15 parts)\end{tabular} &  \begin{tabular}[c]{@{}c@{}}Striping+SA\\(.1, 7 parts)\end{tabular} & \begin{tabular}[c]{@{}c@{}}Striping+SA\\(.1, 15 parts)\end{tabular} &  \begin{tabular}[c]{@{}c@{}}Striping+SA\\(.05, 7 parts)\end{tabular} & \begin{tabular}[c]{@{}c@{}}Striping+SA\\(.05, 15 parts)\end{tabular} \\ \hline
Total cut edges & \textbf{338} & \textbf{666} & 548 & 846 & 452 & 720 & 470 & 748 \\ \hline
\begin{tabular}[c]{@{}c@{}}Max deviation\\ from mean\end{tabular} & 118.6\% & 146.2\% & \textbf{1.8\%} & \textbf{4.4\%} & 5.6\% & 9.2\% & 4.1\% & 4.6\% \\ \hline
Run time (m:s) & 00:01 & 00:02 & 01:03 & 01:17 & 15:50 & 13:25 & 11:44 & 12:10 \\ \hline
\end{tabular}}
\captionof{table}{Performance of various partitions for $k=7,15$ police beats in South Fulton City constructed using $k$-means, combinatorial striping (with balance parameter $\varepsilon=0.1$) and combinatorial simulated annealing using the combinatorial striping method as a warm start  (with varying balance parameter $\varepsilon=0.05$ or $0.1$).}
\label{tab:computations-police-deterministic}
\end{figure} 
 
To address over-policing in the police districting problem, we set up the following experiment using our proposed methods. In the hexagonal grid graph, each vertex was assigned a workload, ranging between 0 and 2947885 minutes/year, representing historical data on work done in that zone (see Figure \ref{img:South Fulton quantities}(a) in the appendix). We first applied the deterministic algorithm, {\sc Dynamic Partition}, to create $k=7,15$ parts with balance parameter $\varepsilon=0.1$. This partition with deterministic workloads (depicted in Figure~\ref{fig:police-deterministic}, Appendix~\ref{append:south-fulton}) was balanced and more compact than the city's original plans and we summarize its performance in Table \ref{tab:computations-police-deterministic}. The combinatorial SA with a striping warm start improved the number of cut edges in the warm start by up to 17\% while maintaining strong balance across part weights, even with $\varepsilon=.1$. Note that in SA a small relaxation in balance $\varepsilon=.05$ to $.1$ yields an approximately 25\% greater improvement in compactness.

However, balancing historic workloads is not sufficient when historic data may not have been collected uniformly \cite{lum2016predict}. To reduce over-policing of parts of the South Fulton City, we set up the following stochastic experiment: Given observed workloads $W_v$ for each $v\in V$, as in Section \ref{sec:synthetic-stochastic}, we generated a random GEV distribution $G_v$ for each $v\in V$, and then discretized each $G_v$ to a distribution $X_v$. The additional scaling factor $s_v$ that is used in determining $X_v$ is chosen to be proportional to the observed workload $W_v$. True fixed, true stochastic, empirical fixed, and empirical stochastic partitions were computed as before for varying sizes of empirical samples. Similarly, we generate 1000 sets of random values $w(v)$ for each $v$ and compute the normalized maximum part difference of each partition.

Finally, to evaluate the balance of each partition method, a weighted average of these differences over 100 sets of true distributions $\{X_v\}$. The weighted average of the normalized maximum part differences of each partition was then computed, with the partitions from each set of distributions weighted by the probability $\prod_{v\in V} \Pr[X_v = W_v]$ that the true distributions $\{X_v\}$ yield the observed workloads. The results of this experiment are plotted in Figure~\ref{fig:stochastic_gev}. Note that again {\sc Stochastic Partition} outperforms {\sc Dynamic Partition} with mean weights as input. Indeed, for larger empirical samples, it achieves better balance for larger empirical samples than {\sc Dynamic Partition} using the true means, which demonstrates superior robustness of {\sc Stochastic} to discrepancies between the observed past workload and the true base workload. 

\section{Conclusion}
In this work, we presented the first polynomial time algorithm for creating balanced, compact and contiguous parts of planar grid graphs. Our paper gives a practical method with guarantees for balanced compact districting on cases that might arise in diverse real-world applications. Traditional approaches using integer programs, for e.g., the flow contiguity formulation of Shirabe~\cite{shirabe2009districting}, are not scalable to real-world data sets. On the other hand, our real-world case studies in Section~\ref{sec:case_study} clearly demonstrate that the combinatorial striping algorithm can serve as a good warm-start (computationally efficient and improves performance) for the commonly used simulated annealing approaches. Initial districting designs based on this work were in fact implemented by the City of South Fulton \cite{zhu2020data}, thus exemplifying the importance of bridging stochastic models with provable guarantees for fair division that are meaningful in practice. 

\section*{Acknowledgement}

The work of Shixiang Zhu and Yao Xie is supported by an NSF CAREER Award CCF-1650913, and NSF CMMI-2015787, DMS-1938106, DMS-1830210.

\bibliographystyle{spmpsci}      
\bibliography{bibliography.bib}   

\appendix
\normalsize
\newpage


\section{Omitted Proofs}\label{app:proofs}

\noindent 
\textbf{Proof of Theorem \ref{thm:unweightedstriping}}

\begin{lemma}\label{lem:algebra}
The approximation ratio between $\frac{\frac{a+1}{a} \OPT_{B}  + \phi(a/A+\phi)(m+n) +  2\phi a +\frac{mn}{A}}{\OPT}$ (expression ($*$) in the proof of Theorem~\ref{thm:unweightedstriping}) and $\OPT$ (the minimum number of cut edges in a partition satisfying the conditions of Theorem~\ref{thm:unweightedstriping}) is bounded by 15.25 and approaches 1 asymptotically as $a,m/a,n/a\rightarrow \infty.$
\end{lemma}

\begin{proof}
We have that
\begin{align*}
&\quad \frac{\frac{a+1}{a} \OPT_{B}  + \phi(a/A+\phi)(m+n) +  2\phi a +\frac{mn}{A}}{\OPT} \\
&\leq \frac{A+1}{A}\cdot\frac{d}{d-2} \frac{\frac{a+1}{a} \OPT_{B}  + \phi(a/A+\phi)(m+n) +  2\phi a +\frac{mn}{A}}{\OPT_{B}} \\
&\leq \frac{A+1}{A}\cdot\frac{d}{d-2} \left(\frac{a+1}{a} + \frac{\phi(a/A+\phi)(m+n) +  2\phi a +\frac{mn}{A}}{\OPT_{B}}\right) \\
&\leq \frac{A+1}{A}\cdot\frac{d}{d-2} \left(\frac{a+1}{a} + \frac{\phi(a/A+\phi)(m+n) +  2\phi a +\frac{mn}{A}}{\frac{mn}{a}}\right) \\
&\leq \frac{A+1}{A}\cdot\frac{d}{d-2} \left(\frac{a+1}{a} + \frac{\phi\left(\frac{1}{m}+\frac{1}{n}+\phi\frac{1}{na}+\phi\frac{1}{ma} +2\phi\frac{a^2}{mn}+\frac{a}{A}\right)(m+n) +\frac{a}{A}}{\frac{mn(a+1)}{a^2}}\right) \\
&\leq \frac{A+1}{A}\cdot\frac{d}{d-2} \left(\frac{a+1}{a} + \phi\left(\frac{1}{m}+\frac{1}{n}+\phi\frac{a}{m}+\phi\frac{a}{n} +2\frac{1}{mn}\right) +\frac{a}{A}\right)
\end{align*}
Asymptotically, as $a\rightarrow \infty$ and $m/a, n/a \rightarrow \infty$, this quantity approaches 1. 

Moreover, in all instances where $d\geq 3$, $a\geq 3$, we have that
\begin{align*}
&\quad \frac{A+1}{A}\cdot\frac{d}{d-2} \left(\frac{a+1}{a} + \phi\left(\frac{1}{m}+\frac{1}{n}+\phi\frac{a}{m}+\phi\frac{a}{n} +2\frac{1}{mn}\right) +\frac{a}{A}\right)\\
&\leq \frac{A+1}{A}\cdot\frac{d}{d-2} \left(\frac{4}{3} + \phi\left(\frac{1}{3d}+\frac{1}{3d}+\phi\frac{1}{d}+\phi\frac{1}{d} +2\frac{1}{81}\right) +\frac{1}{3}\right) \\
&\leq  \frac{A+1}{A}\cdot\frac{d}{d-2}\cdot\frac{5+2\phi/27}{3} + \frac{A+1}{A}\cdot\frac{2\phi/3+2\phi^2}{d-2}\\
&\leq  \frac{10}{9}\cdot1\cdot\frac{5+2\phi/27}{3} + \frac{10}{9}\cdot\frac{2\phi/3+2\phi^2}{1}\\
&\leq 15.25.
\end{align*}
\end{proof}

\noindent 
\textbf{Proof of Theorem \ref{thm:exactunweightedstriping}}

To analyze $\phi$-{\sc Cautious Striping} in this case, we first prove a lemma which is used in bounding the perimeter of the parts in the strips of height $a$ and $a+1$.\footnote{Note that this yields a streamlined proof of Theorem 5 of~\cite{christou1996optimal}.} The hypothesis that $k|mn$ allows us to obtain a tighter bound on the approximation ratio than Theorem~\ref{thm:unweightedstriping}.

\begin{lemma}\label{lem:uniformstriping} Let $A\in\mathbb{N}$ and $a=\lfloor\sqrt{A}\rfloor.$ Then if $\sqrt{A}-a>0$, for each $h\in \{a,a+1\}$ we have that
\[
\frac{A}{h} + h \leq \lceil2\sqrt{A}\rceil.
\]
\end{lemma}
\begin{proof} There are four cases dependent on the stripe height and $\sqrt{A}-a$:

1. $h=a$, $0<\sqrt{A}-a <\frac{1}{2}$. Then $\lceil2\sqrt{A}\rceil=2a+1$. Note that $A \leq a\left(a+1\right)$, since $A<\left(a+\frac{1}{2}\right)^2=a\left(a+1\right)+\frac{1}{4}$ and both $a(a+1)$ and $A$ are integers. Therefore $A/a \leq a+1$ and adding $a$ to both sides yields the desired inequality.

2. $h=a$, $\frac{1}{2}<\sqrt{A}-a <1$. Then $\lceil2\sqrt{A}\rceil=2a+2$. Note that $A \leq a\left(a+2\right)$, since $A<\left(a+1\right)^2$ and $a\left( a+2\right)=\left(a+1\right)^2-1$ and $A$ is an integer. Therefore $A/a \leq a+2$ and adding $a$ to both sides yields the desired inequality.

3. $h=a+1$, $0<\sqrt{A}-a < \frac{1}{2}$. Then $\lceil2\sqrt{A}\rceil=2a+1$. As in case 1, $A\leq a\left(a+1\right)$, so $A/\left(a+1\right) \leq a$ and adding $a+1$ to both sides yields the desired inequality.

4. $h=a+1$, $\frac{1}{2}< \sqrt{A}-a <1$. Then $\lceil2\sqrt{A}\rceil=2a+2$. Since $A < \left(a+1\right)^2$, $A/\left(a+1\right) \leq \left(a+1\right)$ and adding $a+1$ to both sides yields the desired inequality.
\end{proof}

\begin{proof}[Proof of Theorem~\ref{thm:exactunweightedstriping}] Balance is immediate by the construction of $\phi$-{\sc Cautious Striping}, and contiguity follows as in the proof of Theorem~\ref{thm:unweightedstriping}.

Recall that the smallest possible perimeter for a region containing $A$ vertices is $2\lceil 2\sqrt{A}\rceil$~\cite{christou1996optimal}.


First, we consider the strips not in $R$. Consider the parts in a strip of height $h$. There are two cases. 

Suppose $a^2=A$. If $h=a$, then all the parts are $a\times a$ squares, which have minimum perimeter. If $h=a+1$, then each part in the strip will fit in a rectangle with $a+1$ rows and $a$ columns, because they cannot occupy more than $a-1$ complete columns and if they do so, cannot contain vertices in more than one other column. Therefore each part has perimeter $4a+2$, or $1+\frac{1}{2a}= 1+\frac{1}{\lceil2\sqrt{A}\rceil}$ times that of the optimum $4a$.

Otherwise, let $p$ denote the number of parts that lie in the strip. The union of these parts is a rectangle with $h$ rows and $\lceil pA/h \rceil$ columns, with up to the bottom $h-1$ vertices of the last column removed. The $p-1$ borders between parts each\ have length at most $h+1$, and each of these borders is shared by two parts. Therefore the total perimeter of the $p$ parts is 
\begin{align*}
2h + 2\lceil pA/h \rceil+ 2(p-1)(h+1) &\leq 2pA/h + 2p(h+1) +2 \\
&\leq 2p(A/h+h+1) + 2 \\
&\leq 2p\left(1+\frac{1}{\lceil2\sqrt{A}\rceil}\right)\lceil2\sqrt{A}\rceil + 2,
\end{align*}
where the last inequality follows by Lemma~\ref{lem:uniformstriping}.

We next bound the perimeter of each part in $R$. At most $\frac{\lceil \phi \sqrt{A}\rceil}{n}$ proportion of the vertices of $G$ will be in $R$. Let $c_1$ be the number of complete columns remaining in the strips of height $a$ and $c_2$  be the number of complete columns remaining in the strips of height $a+1$ after step 9 of the algorithm. These parts are divided into three types:
\begin{enumerate}
    \item Parts contained in rows $1$ through $s_1a$: These parts fit in a box of width $c_1+1$ and height $\lceil A/c_1\rceil+1$. The right edge of each part is straight, containing at most $\lceil  A/c_i\rceil+1$ vertical edges. The top and bottom edges each contain at most $c_i+1$ horizontal edges and one vertical edge. The left edge contains at most $\lceil A/c_i\rceil+1$ vertical edges and at most four horizontal edges, depending on whether the boundary with the strips of height $h$ is a straight line. Therefore, the total perimeter of a part is at most
    \[
    2c_i+ 2\lceil A/c_i\rceil + 10 \leq 2c_i+ 2\lfloor A/c_i\rfloor + 12.
    \]
    Since $\lfloor (\phi-1)a \rfloor \leq c_i\leq \lceil \phi a \rceil$, the term $2c_i+ 2\lceil A/c_i\rceil$ is maximized when $c_i$ takes one of its two possible extreme values. This maximum value is at most $2\left(\phi+\frac{1}{\phi}\right)\lceil \sqrt{A}\rceil+15$. 

Therefore, the modified striping approach achieves a (asymptotic) constant factor approximation for compactness, where the constant is $\alpha \approx (\phi+\frac{1}{\phi})$. 
    Therefore, the relative error to an optimum single part containing $A$ vertices is at most
    \[
    \frac{2\left(\phi+\frac{1}{\phi}\right)\lceil\sqrt{A}\rceil+O(1) - 2\lceil 2\sqrt{A}\rceil }{2\lceil 2\sqrt{A}\rceil} =   \left(\frac{\sqrt{5}}{2}-1\right)+\frac{O(1)}{\sqrt{A}}.
    \]
    \item At most one part has cells in both row $s_1a$ and $s_1a+1$. This part has all sides of length at most $2\phi a$, for total perimeter $O(\sqrt{A})$ and relative error $O(1)$.
    \item Parts contained in rows $s_1a+1$ through $m$: by the argument for the parts contained in rows $1$ through $s_1a$, these parts each have relative error $\left(\frac{\sqrt{5}}{2}-1\right)+\frac{O(1)}{\sqrt{A}}$ also.
\end{enumerate}
Averaging over all parts, we have that the total relative error is at most
\begin{align*}
&\frac{1}{\lceil2 \sqrt{A}\rceil} \left(\frac{n-|R|}{n}\right)+\left(\frac{\sqrt{5}}{2}-1+\frac{O(1)}{\sqrt{A}}\right)\left(\frac{|R|}{n}\right) + \frac{O(1)}{n}\\
&\leq
\frac{1}{\lceil2 \sqrt{A}\rceil}+\frac{\lceil \phi \sqrt{A}\rceil \cdot\left(\frac{\sqrt{5}}{2}-1\right)+O(1)}{n}.
\end{align*}
\end{proof}

\noindent 
\textbf{Proof of Theorem \ref{thm:variance_bound}}

\begin{proof} By the triangle inequality, we have that 
\begin{align*}
\E\left[\sum_{i=1}^k \left(\sum_{v\in V_i} X_v - A\right)^2\right] \leq \E\left[\sum_{i=1}^k \left(\sum_{v\in V_i} X_v - \sum_{v\in V_i} \mu(v)\right)^2\right]\\
\hspace{-2.2cm}+ \left| \E \left[\sum_{i=1}^k \left(\sum_{v\in V_i} X_v - A\right)^2 - \sum_{i=1}^k \left(\sum_{v\in V_i} X_v - \sum_{v\in V_i} \mu(v)\right)^2   \right] \right|.
\end{align*}
We bound each of the terms on the right hand side of this equation separately.

In $\E\left[\sum_{i=1}^k \left(\sum_{v\in V_i} X_v - \sum_{v\in V_i} \mu(v)\right)^2\right]$, note that each of the summands $\left(\sum_{v\in V_i} X_v - \sum_{v\in V_i} \mu(v)\right)^2$ is the variance of the distribution $\sum_{v\in V_i} X_v$. This distribution is the sum of the distributions $X_v$ and is hence a distribution with mean $\sum_{v\in V_i} \mu(v)$ and variance $c\sum_{v\in V_i} \mu(v)$. Therefore, 
\[
\E\left[ \left(\sum_{v\in V_i} X_v - \sum_{v\in V_i} \mu(v)\right)^2\right] = c\sum_{v\in V_i} \mu(v),
\]
and summing over all parts $V_i$ gives 

$$\E\left[\sum_{i=1}^k \left(\sum_{v\in V_i} X_v - \sum_{v\in V_i} \mu(v)\right)^2\right] \leq c\sum_{v\in V} \mu(v) = ckA.$$

Second, we have that
\begin{align*}
    &\quad \left| \E\left[\sum_{i=1}^k \left(\sum_{v\in V_i} X_v - A\right)^2 - \sum_{i=1}^k \left(\sum_{v\in V_i} X_v - \sum_{v\in V_i} \mu(v)\right)^2   \right] \right| \\
    & = \E\left[\left| \sum_{i=1}^k \left(\sum_{v\in V_i} X_v - A\right)^2 - \sum_{i=1}^k \left(\sum_{v\in V_i} X_v - \sum_{v\in V_i} \mu(v)\right)^2  \right| \right] \\
    & \leq \E\left[\sum_{i=1}^k \left| \left(\sum_{v\in V_i} X_v - A\right)^2 - \left(\sum_{v\in V_i} X_v - \sum_{v\in V_i} \mu(v)\right)^2  \right| \right] \\
    & \leq  \E\left[\sum_{i=1}^k \left| 2\sum_{v\in V_i} X_v \left(\sum_{v\in V_i} \mu(v)- A\right) + A^2 - \left(\sum_{v\in V_i} X_v\right)^2  \right| \right]\\
    & \leq  \E\left[\sum_{i=1}^k 2\sum_{v\in V_i}X_v \left|  \sum_{v\in V_i} \mu(v)- A\right| + \left|A^2 - \left(\sum_{v\in V_i} X_v\right)^2  \right| \right]\\
    & \leq  \E\left[\sum_{i=1}^k 2\sum_{v\in V_i}X_v \varepsilon A + (2\varepsilon+\varepsilon^2) A^2 \right]\\
    & =  \sum_{i=1}^k \left(2\sum_{v\in V_i}\E [X_v] \varepsilon A + (2\varepsilon+\varepsilon^2) A^2\right)\\
    & = \sum_{i=1}^k \left(2\sum_{v\in V_i}\mu(v) \varepsilon A + (2\varepsilon+\varepsilon^2) A^2\right) \\
    & = 2\sum_{v\in V} \mu(v) \varepsilon A + k(2\varepsilon+\varepsilon^2) A^2 \\
    & = (4\varepsilon + \varepsilon^2) k A^2.
\end{align*}

Therefore $\E\left[\sum_{i=1}^k \left(\sum_{v\in V_i} X_v - A\right)^2\right] \leq (4\varepsilon + \varepsilon^2) k A^2 + ckA$ as desired.
\end{proof}


\noindent
\textbf{Proof for Theorem \ref{thm:load_balancing}}

We adapt the proof of Kleinberg et al.\ for their algorithm for the stochastic load-balancing problem~\cite{kleinbergbursty}. In this problem, random variables $X_1,\dots,X_n$ are assigned to $k$ bins $V_1,\dots,V_k$, with the objective being to minimize the expected maximum bin value $\E[ \max_{1\leq j\leq k} \sum_{v \in V_j} X_v]$. 

The algorithm repeats the following series of steps until the solution is obtained. After each iteration, the variables $X_v$ are halved.
\begin{enumerate}
    \item If the total value $\sum_{v=1}^{n}X_v\cdot \mathds{1}_{\{X_v> 1\}}$ of the exceptional parts of the the variables is greater than $1$, move to the next iteration. 
    \item Assign each variable $X_v$ a weight $\beta_{1/k}(N_v)= \log \E[k^{N_v}]/\log k$, where $N_v = X_v\cdot \mathds{1}_{\{X_v\leq 1\}}$. The $X_v$ are taken one-by-one and placed in the bin with smallest total weight. If it would be impossible to place an item without a bin exceeding total weight of 18, move to the next iteration.
    \item If all $X_v$ were assigned to bins, each bin having total weight at most 18, return this assignment.
\end{enumerate}
We refer to Steps 2 and 3 as the \textit{greedy bin-packing method}. Kleinberg et al.\ show that this method is an $O(1)$-approximation algorithm for the expected maximum bin value. Before addressing {\sc Stochastic Partition}, we first summarize this proof, which has two parts: upper-bounding the value of the output of their algorithm, and lower-bounding the value of the optimum solution.\newline

\noindent\textit{Upper bound.} This bound on the objective follows from the upper bounds in steps 1 and 2 of the algorithm, which are satisfied by the feasible iteration, and is summarized in the following lemma:

\begin{lemma}[Lemma 3.7 of~\cite{kleinbergbursty}]\label{lem:stochasticupperbound}
Let $\{V_1,\dots,V_k\}$ be the bins of variables produced by the stochastic load-balancing algorithm. Then $\E[ \max_{1\leq i\leq k} \sum_{v\in V_i}N_{v}] =O(1)$, where $N_{v}=X_v\cdot \mathds{1}_{\{X_v\leq 1\}}$.
\end{lemma}

To show this lemma, a probabilistic bound using Markov's inequality is applied to the quantity $\E[ \max_{1\leq j\leq k} \sum_{v \in V_j} N_v]$, the objective value considering only the ordinary parts of the $X_v$. Finally, since the total expected value $\sum_{v=1}^{n}S_{v}$ of the exceptional parts of the $X_v$ is at most 1, the expected maximum bin value can be bounded:
\begin{align*}
\E\left[ \max_{1\leq i\leq k} \sum_{v\in V_i}X_{v}\right] &= \E\left[ \max_{1\leq i\leq k} \sum_{v\in V_i}N_{v} + S_{v}\right]\\
&\leq  \E\left[ \max_{1\leq i\leq k} \sum_{v\in V_i}N_{v}\right] + \E\left[ \max_{1\leq i\leq k} \sum_{v\in V_i}S_{v}\right] = O(1).
\end{align*}

\noindent\textit{Lower bound.} If an iteration of the algorithm is infeasible, then they show that the optimal solution has a corresponding lower bound. Two preliminary lemmas are used to show that if the total weights of the items is large, then the expected maximum part size for all assignments (in particular for the optimum assignment) is bounded below by a constant:

\begin{lemma}[Lemmas 3.3 and 3.10 of~\cite{kleinbergbursty}]\label{lem:kleinberglowerboundexceptional}
Let $X_1,\dots,X_n$ be independent non-negative random variables and let $L>0$. Suppose each $X_v$ has support contained in $\{0\} \cup [L,\infty)$ and $\sum_{v=1}^{n} \E[X_v]\geq L$. Then for all assignments of $X_1,\dots,X_n$ to bins $V_1,\dots,V_k$, we have that $\E[\max_{1\leq i\leq k} \sum_{v\in V_i} X_v] \geq L/4$.
\end{lemma}

\begin{lemma}[Lemmas 3.4 and 3.10 of~\cite{kleinbergbursty}]\label{lem:kleinberglowerbound}
Let $X_1,\dots,X_n$ be independent non-negative random variables bounded above by 1. For a given $k\in \mathbb{N}$, suppose that $\sum_{v=1}^{n} \log \E[k^{X_v}]/\log k \geq 17k$. Then for all assignments of $X_1,\dots,X_n$ to bins $V_1,\dots,V_k$, we have that $\E[ \max_{1\leq i\leq k} \sum_{j\in V_i} X_j] =\Omega(1)$.
\end{lemma}

Applying Lemma~\ref{lem:kleinberglowerboundexceptional} to the exceptional parts $S_v = X_v\cdot \mathds{1}_{\{X_v> 1\}}$ and using $L=1$, we have that for any iteration of the stochastic load-balancing algorithm which terminates in step 1, any assignment of variables to bins will have expected maximum part weight $\Omega(1)$. On the other hand, suppose an iteration of the stochastic load-balancing algorithm terminates in step 2. Note that all weights $ \log \E[k^{N_v}]/\log k$ are at most 1. Therefore, all $k$ bins must have had total weight at least 17 when the greedy bin-packing algorithm stopped, so $\sum_{v=1}^{n} \log \E[k^{X_v}]/\log k\geq \sum_{v=1}^{n} \log \E[k^{N_v}]/\log k \geq 17k$. Then one can apply Lemma~\ref{lem:kleinberglowerbound} and again conclude that any assignment of variables to bins will have expected maximum part weight $\Omega(1)$. 

Combined with the $O(1)$ upper bound above for the result of the algorithm, one can obtain the desired constant-factor approximation. We now show how to extend this argument to the {\sc Stochastic Partition} guarantees. 

\begin{proof}

\noindent\textit{Upper bound.} Note that the proof of the upper bound for the load-balancing case, including Lemma~\ref{lem:stochasticupperbound}, depends only on the assignment satisfying the condition in step 3 of the load-balancing algorithm. In particular, it does not use the particular structure of the greedy bin-packing algorithm, and applies to any assignment such that the total weight of the exceptional parts is at most 1 and no bin $V_i$ has total weight $\sum_{v \in V_i} \beta_{1/k}(N_v)$ of more than 18. Therefore, as {\sc Stochastic Partition} satisfies these criteria by construction, this part of the proof applies in its entirety to the graph partitioning setting, and the objective value of the partition returned by {\sc Stochastic Partition} is $O(1)$.\newline

\noindent\textit{Lower bound.} If the partitioning algorithm is infeasible (it cannot find a consistent partition that satisfies the balance constraint) then we wish to show that the optimal solution has a corresponding lower bound. To be able to apply Lemma~\ref{lem:kleinberglowerbound} in the graph partitioning setting, we first need to show that if the partitioning algorithm {\sc Dynamic Partition} is infeasible in an instance, then the condition  $\sum_{i=1}^{n} \log \E[k^{Y_i}]/\log k \geq 17k$ is satisfied.

\begin{lemma}\label{lem:stochasticlowerbound}
Let $G=(V,E)$ be a graph and let $0\leq w(v)\leq 1$ for all $v\in V$ be vertex weights. Let $k\in \mathbb{N}$ be a desired number of parts, and let $v_1,\dots,v_n$ be an ordering of $V$. Suppose there is no partition of $G$ into $k$ parts, each with total weight at most $C$, that is consistent with this ordering. Then $\sum_{v\in V}w(v) \geq (C-1)k$.
\end{lemma}
\begin{proof}[Proof of Lemma~\ref{lem:stochasticlowerbound}]
Let $i$ be the smallest integer such that there is a consistent partition of $\{v_1,\dots,v_i\}$ into $k$ parts, each with total weight at most $C$. If there exists such a partition each of whose $k$ parts has total weight at least $(C-1)$, the result follows. If not, let $\mathcal{P}=\{V_1,\dots,V_k\}$ be the partition which, of all such consistent and balanced partitions, has the largest number of consecutive initial parts $V_1,\dots,V_j$ each having total weight at least $(C-1)$.

Consider part $V_{j+1}$, which has total weight less than $(C-1)$. Let $v_r$ be the last vertex of $V_{j+1}$. Add vertices $v_{r+1},v_{r+2},\dots$ to $V_{j+1}$ (from $V_\ell$ for some $\ell>j+1$ if necessary) until vertex $v_i$ is added or $\sum_{v\in V_{j+1}}w(v)\geq (C-1)$. Note that at all steps in this process, the partition remains consistent and the weight of $V_{j+1}$ never exceeds $C$ since $w(v)\leq 1$ for all $v\in V$. If vertex $v_i$ is added to $V_{j+1}$, this contradicts the minimality of $i$. If vertices are added which make $\sum_{v\in V_{j+1}}w(v)\geq (C-1)$, this contradicts the maximality of $j$ and the choice of $\mathcal{P}$. 

Therefore, there exists a consistent partition of $\{v_1,\dots,v_{i}\}$ each of whose $k$ parts has total weight at least $(C-1)$, and so $\sum_{v\in V}w(v)\geq \sum_{t=1}^{i} v_t \geq (C-1)k$.
\end{proof}

By Theorem~\ref{thm:dynamic}, {\sc Dynamic Partition} will find a partition of $G$ into $k$ parts each with total weight at most $C=18$, and which is consistent with a given ordering, if such a partition exists. Therefore, if it does not find one, $\sum_{v\in V}w(v) \geq 17k$ and we may apply Lemma~\ref{lem:stochasticlowerbound}, obtaining the necessary condition to apply Lemma~\ref{lem:kleinberglowerbound}. In particular, when $i=i^*-1$, where $i^*$ is the number of halving steps used in {\sc Stochastic Partition}, the optimal partition has expected maximum part size $\Omega(1)$. Combined with the $O(1)$ upper bound above for the solution produced by {\sc Stochastic Partition}, we obtain the desired constant-factor approximation.
\end{proof}
\clearpage
\section{Experimental Results}

\subsection{Synthetic results}

Table~\ref{tab:computations-artificial} contains data on the performance of partitions on a $100\times 100$ hexagonal grid created using {\sc Dynamic Partition} and simulated annealing. These partitions appear in Figure~\ref{fig:dynamic_example_weights}.

\begin{center}
\begin{table}[h]
{\footnotesize
\begin{tabular}{|c|c|c|c|c|}
\hline
&  Striping, $\varepsilon = .02$ & Striping + SA, $\varepsilon = .02$ & Striping, $\varepsilon = .05$ & Striping + SA, $\varepsilon = .05$ \\ \hline
\begin{tabular}[c]{@{}c@{}}Total \\ cut edges\end{tabular}  & 3,601 & \textbf{3,570} & 3,528 & \textbf{3,498}  \\ \hline
\begin{tabular}[c]{@{}c@{}}Run time\\ (hour:min:sec)\end{tabular} & 0:26:58  & 2:25:23  & 0:39:22 & 0:58:59 \\ \hline
\end{tabular}%
\caption{Performance of partitions on the hexagonal grid. The total number of cut edges corresponds to our objective function. The second and fourth columns correspond to the results obtained using {\sc Dynamic Partition}. The third and the fifth columns correspond to the results using simulated annealing initialized with the striping algorithm.}
\label{tab:computations-artificial}
}
\end{table}
\end{center}
\subsection{City of South Fulton}
\label{append:south-fulton}
\paragraph{Fire department territories in South Fulton City}

Figure~\ref{fig:SF_fire_map} shows the current station locations and the distribution of the police workload for answering fire calls.

\begin{figure}[h]
    \centering
    \begin{subfigure}[h]{0.3\linewidth}
    \includegraphics[width=\linewidth]{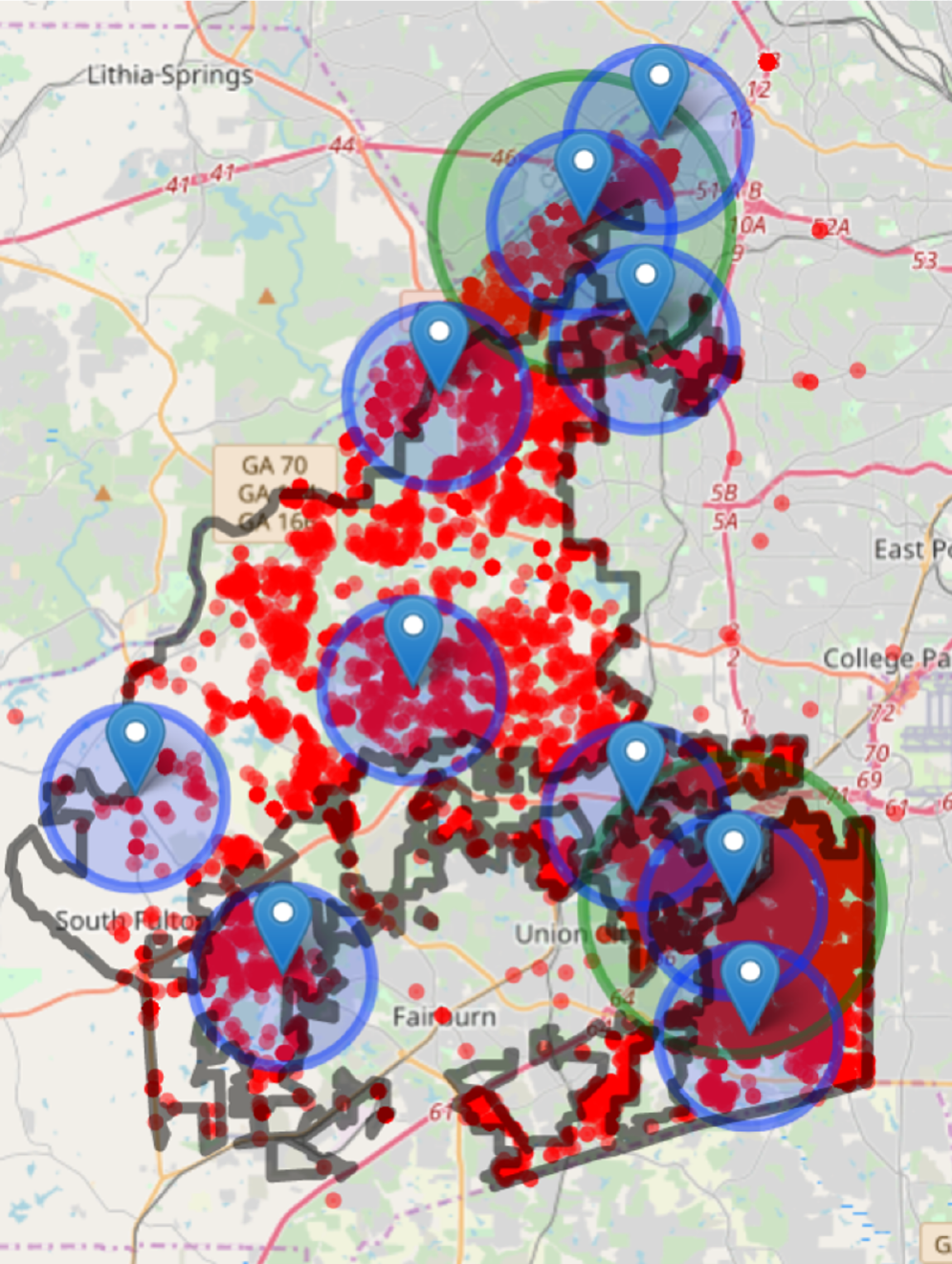}
    \end{subfigure}
    \begin{subfigure}[h]{0.3\linewidth}
    \includegraphics[width=\linewidth]{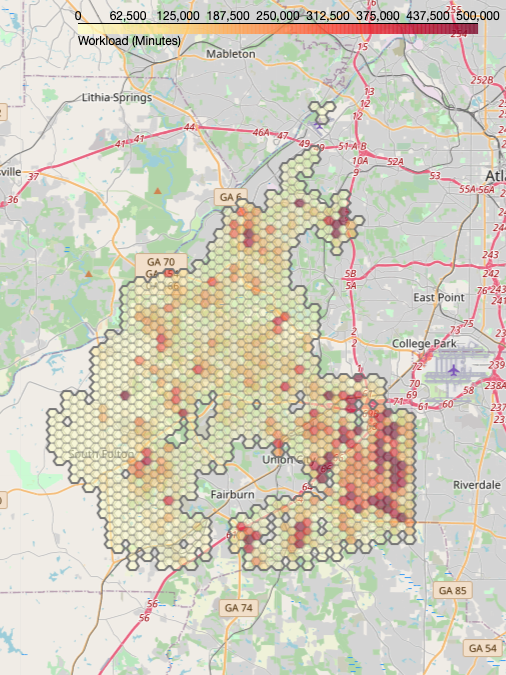}
    \end{subfigure}
    \caption{Fire stations in the City of South Fulton.
    (a) The ten current stations' locations, shown in blue, with coverage radii of 1.5 miles in blue and 2.5 miles in green. Red dots indicate fire incidents in the city;
    (b) The city is divided into 1,236 hexagonal polygons, where the color depth shows the workload (minutes per year).}
    \label{fig:SF_fire_map}
\end{figure}
Figure~\ref{fig:fire-data-appendix} presents the partitions of the fire territories in South Fulton with balance parameter $\varepsilon=.1$, which are generated by weighted $k$-means, the striping algorithm, and the simulated annealing, respectively. Simulated annealing takes the partitions generated by the striping algorithm as its warm start initialization. Here, we consider three sets of weights: when the workloads correspond to sample means and to perturbed workloads that are increased by 0.1 and 0.5. Performance of these partitions is described in Table~\ref{tab:computations-fire-data}.

\begin{figure}[!h]
\centering
\begin{subfigure}[h]{0.22\linewidth}
\includegraphics[width=\linewidth]{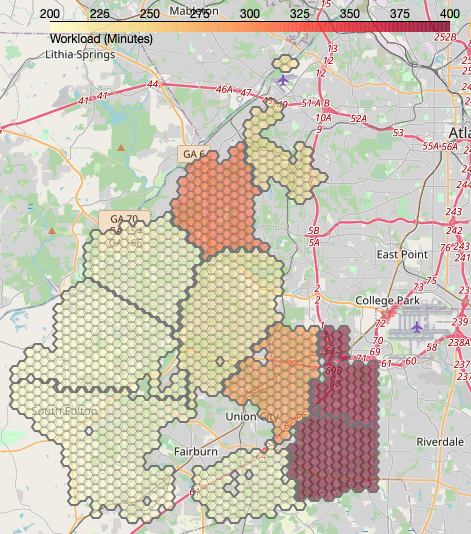}
\caption{$k$-means(.1)}
\end{subfigure}
\hspace{.1in}
\begin{subfigure}[h]{0.22\linewidth}
\includegraphics[width=\linewidth]{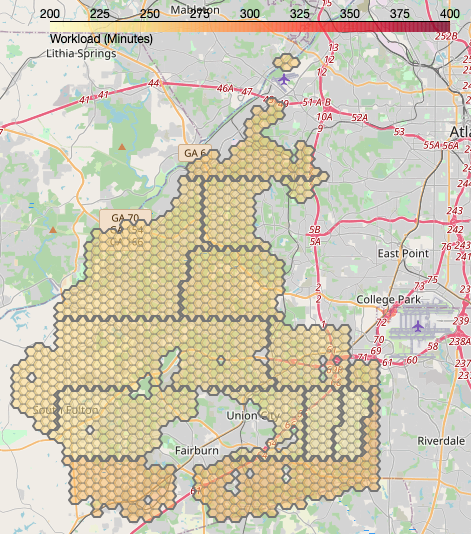}
\caption{Striping(.1)}
\end{subfigure}
\hspace{.1in}
\begin{subfigure}[h]{0.22\linewidth}
\includegraphics[width=\linewidth]{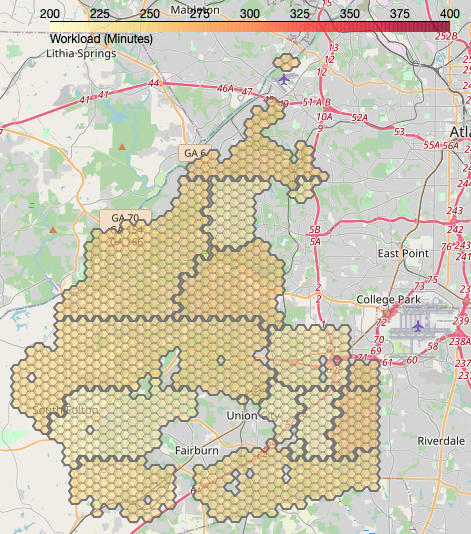}
\caption{Striping+SA(.1)}
\end{subfigure}
\vfill
\begin{subfigure}[h]{0.22\linewidth}
\includegraphics[width=\linewidth]{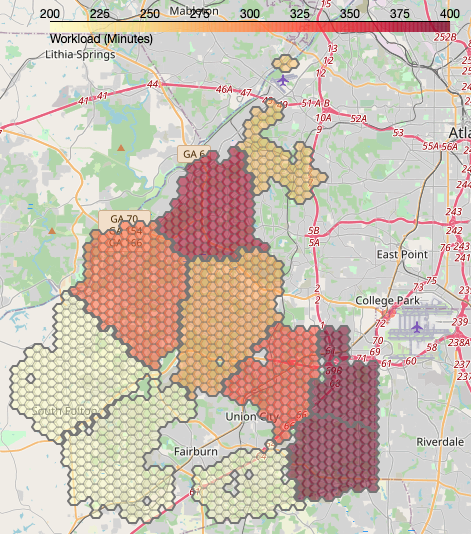}
\caption{$k$-means(.5)}
\end{subfigure}
\hspace{.1in}
\begin{subfigure}[h]{0.22\linewidth}
\includegraphics[width=\linewidth]{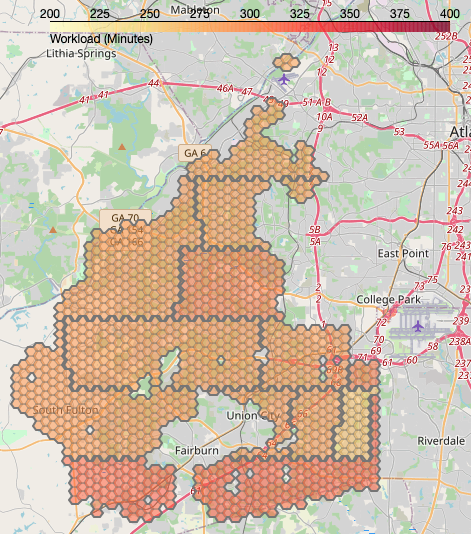}
\caption{Striping(.5)}
\end{subfigure}
\hspace{.1in}
\begin{subfigure}[h]{0.22\linewidth}
\includegraphics[width=\linewidth]{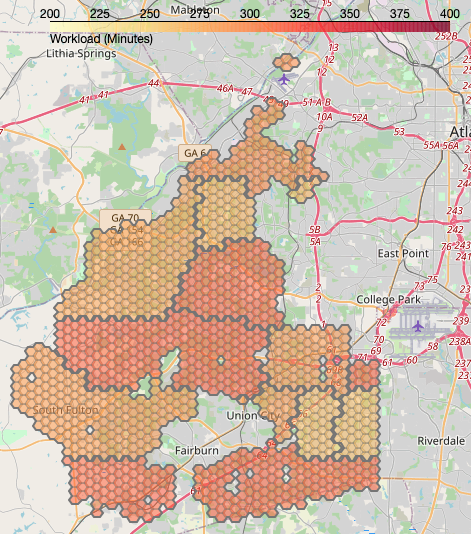}
\caption{Striping+SA(.5)}
\end{subfigure}
\caption{Fire station districting results of the City of South Fulton. The performance metrics of these partitions are summarized Table~\ref{tab:computations-fire-data}. The color depth represents the workload per part. The dark lines outline the part boundaries. 
(b, e) correspond to partitions with 10 beats generated by the striping algorithm with balance parameter $\varepsilon=.1$; 
(c, f) correspond to partitions with 10 beats generated by simulated annealing algorithm with balance parameter $\varepsilon=.1$, where (b, e) are taken as their warm start initialization.}
\label{fig:fire-data-appendix}
\end{figure}

\paragraph{Partitions for police districting}
Figure~\ref{img:South Fulton quantities} shows five different statistics estimated from 911 police call data provided by the South Fulton Police Department from 2018 to 2019. 

Figure~\ref{fig:police-deterministic} shows two types of partitions (7-beat and 15-beat) for the police beats configuration in South Fulton using weighted $k$-means, the striping algorithm, and simulated annealing, respectively. 
Simulated annealing uses the partitions generated by the striping algorithm as its warm start initialization, with balance parameters $\varepsilon=0.1$ and $\varepsilon=0.05$. Performance of these partitions is described in Table~\ref{tab:computations-police-deterministic}.
\begin{figure}[h]
\centering
\begin{subfigure}[h]{0.19\linewidth}
\includegraphics[width=\linewidth]{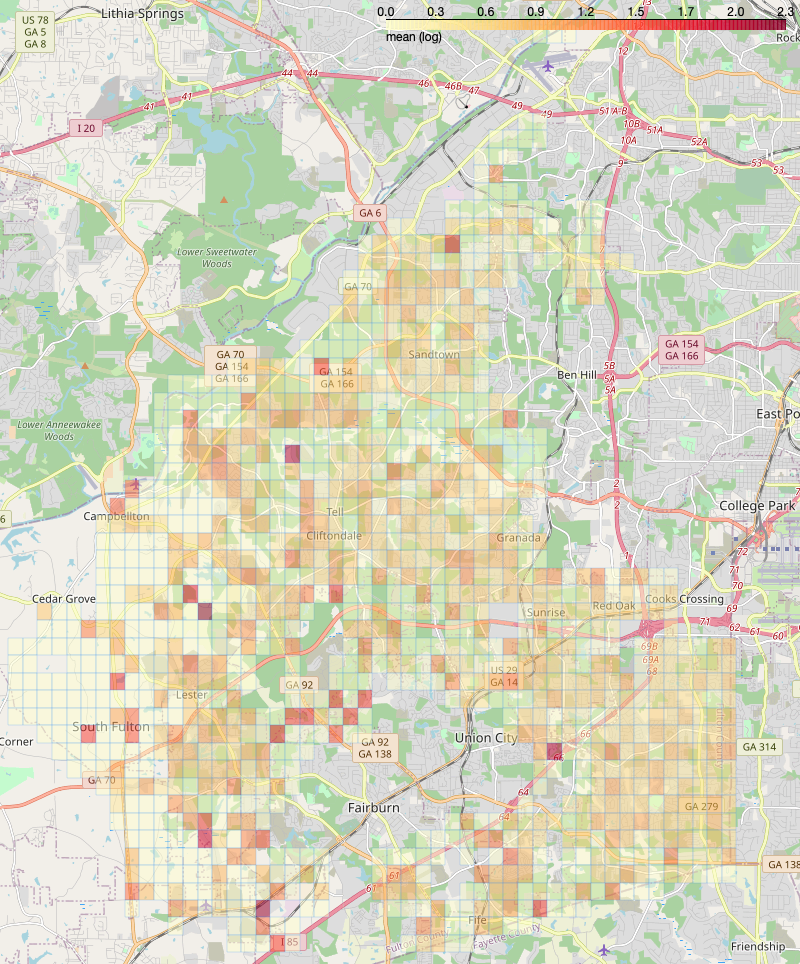}
\caption{}
\end{subfigure}
\begin{subfigure}[h]{0.19\linewidth}
\includegraphics[width=\linewidth]{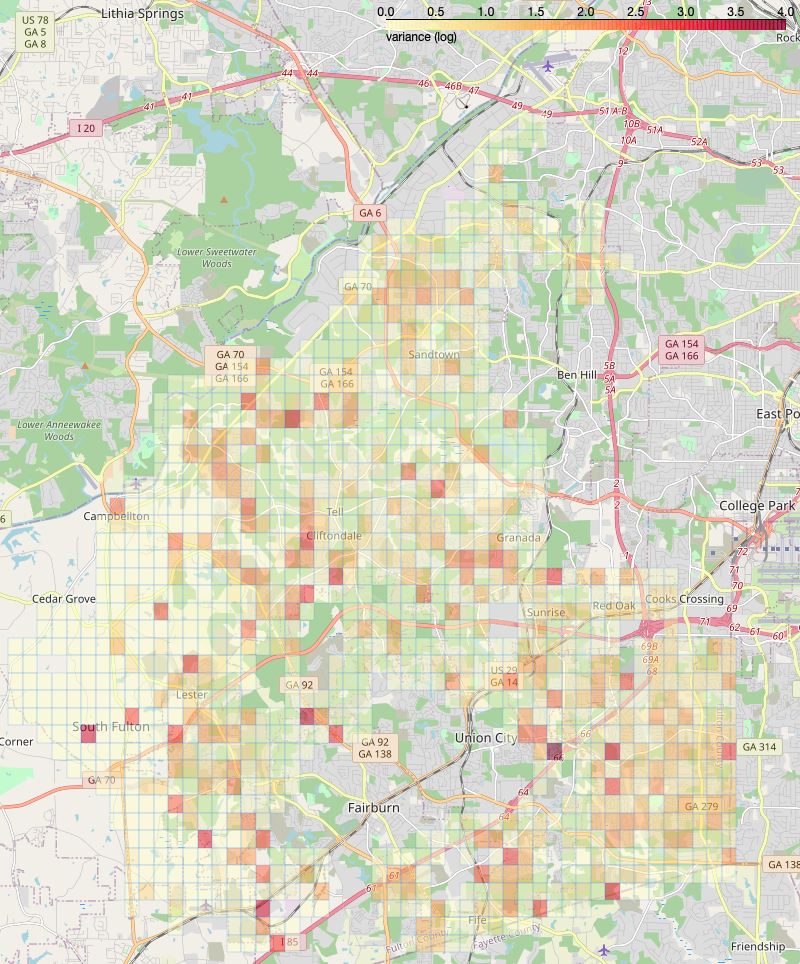}
\caption{}
\end{subfigure}
\begin{subfigure}[h]{0.19\linewidth}
\includegraphics[width=\linewidth]{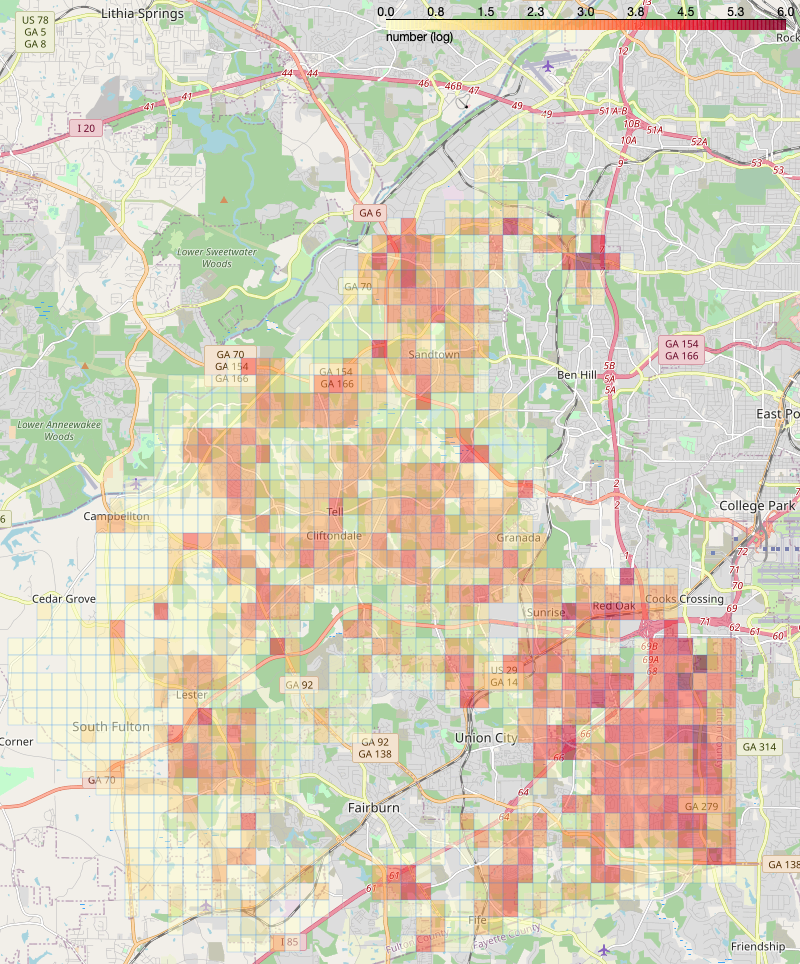}
\caption{}
\end{subfigure}
\begin{subfigure}[h]{0.19\linewidth}
\includegraphics[width=\linewidth]{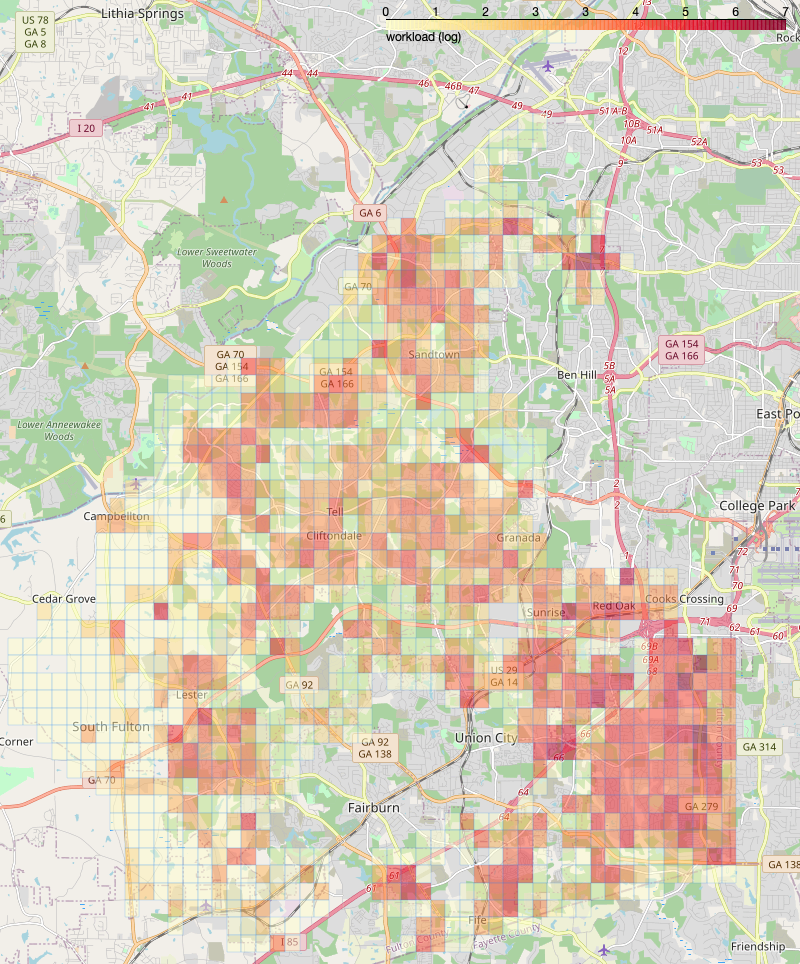}
\caption{\ }
\end{subfigure}
\begin{subfigure}[h]{0.19\linewidth}
\includegraphics[width=\linewidth]{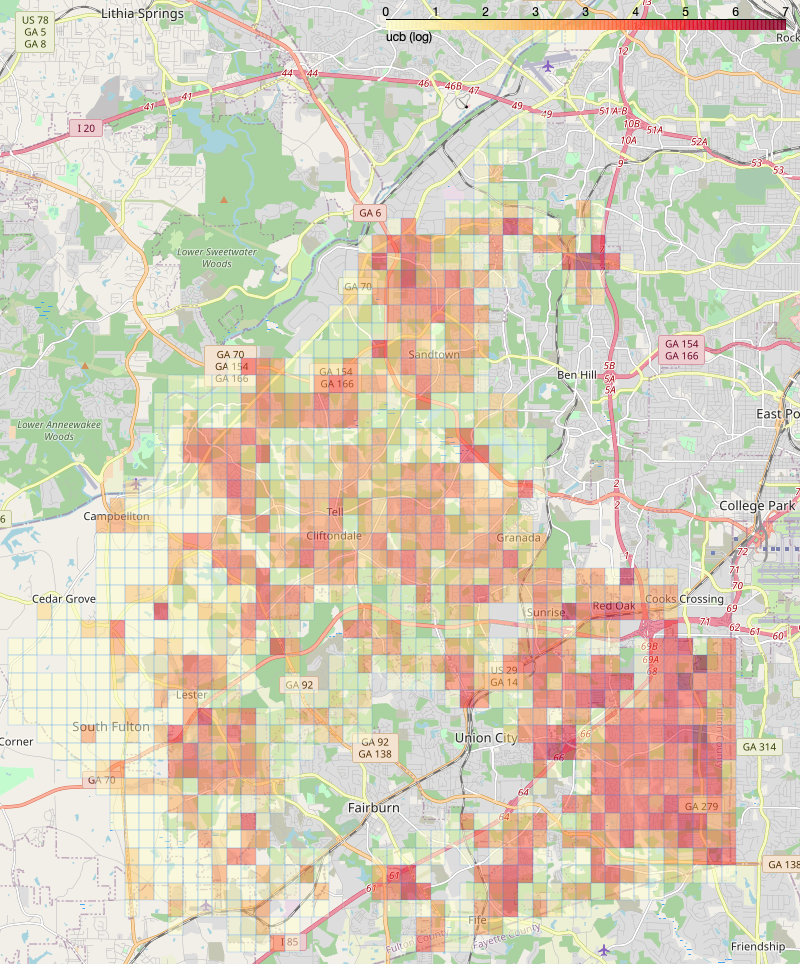}
\caption{\ }
\end{subfigure}
\caption{Statistics for the 911 call data in each grid square in South Fulton. In each plot, red represents the largest quantity, white the least. (a): Mean length of response time to 911 calls. (b): Variance in response time to 911 calls. (c): Number of 911 calls. (d): Total workload (number of calls times mean response time). (e): Upper bound of 95\% confidence interval for response time (where response time is assumed to follow a Gaussian distribution).}
\label{img:South Fulton quantities}
\begin{subfigure}[h]{0.22\linewidth}
\includegraphics[width=\linewidth]{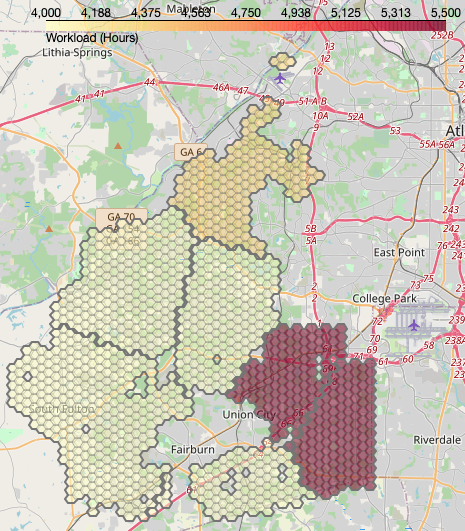}
\caption{$k$-means}
\end{subfigure}
\hspace{.1in}
\begin{subfigure}[h]{0.22\linewidth}
\includegraphics[width=\linewidth]{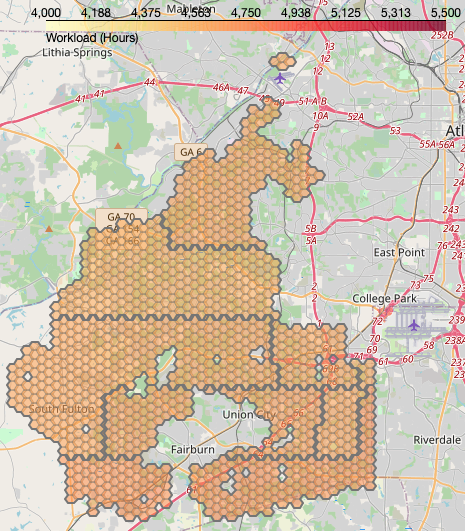}
\caption{Striping}
\end{subfigure}
\hspace{.1in}
\begin{subfigure}[h]{0.22\linewidth}
\includegraphics[width=\linewidth]{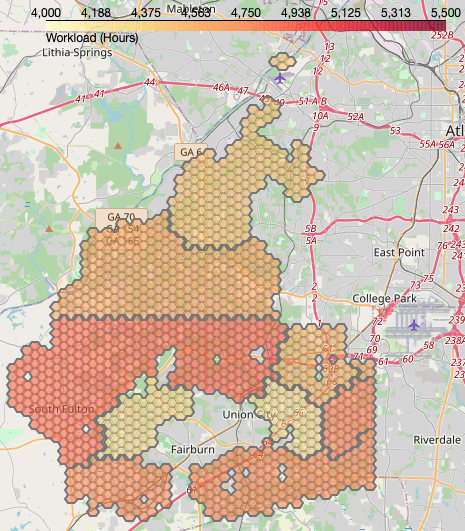}
\caption{Striping+SA(.1)}
\end{subfigure}
\hspace{.1in}
\begin{subfigure}[h]{0.22\linewidth}
\includegraphics[width=\linewidth]{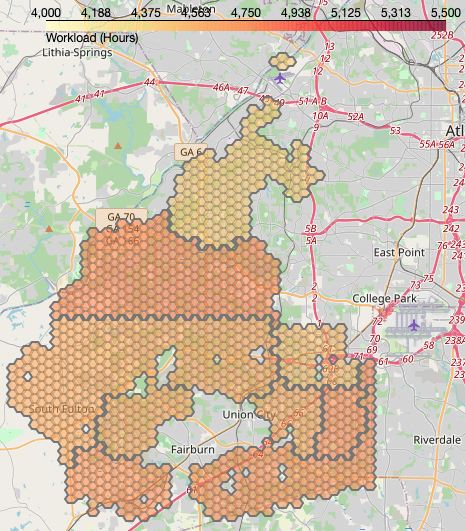}
\caption{Striping+SA(.05)}
\end{subfigure}
\vfill
\begin{subfigure}[h]{0.22\linewidth}
\includegraphics[width=\linewidth]{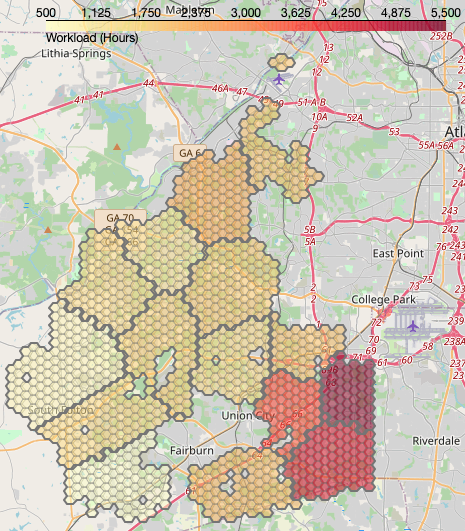}
\caption{$k$-means}
\end{subfigure}
\hspace{.1in}
\begin{subfigure}[h]{0.22\linewidth}
\includegraphics[width=\linewidth]{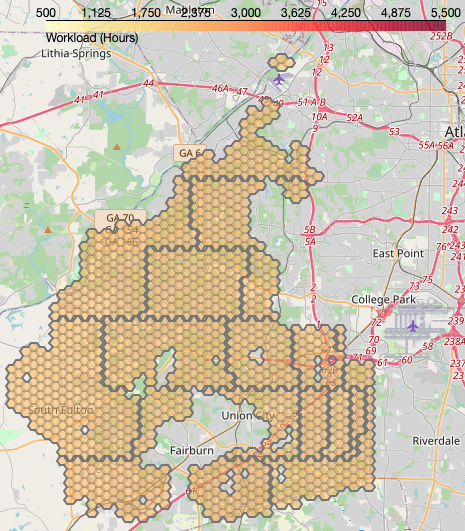}
\caption{Striping}
\end{subfigure}
\hspace{.1in}
\begin{subfigure}[h]{0.22\linewidth}
\includegraphics[width=\linewidth]{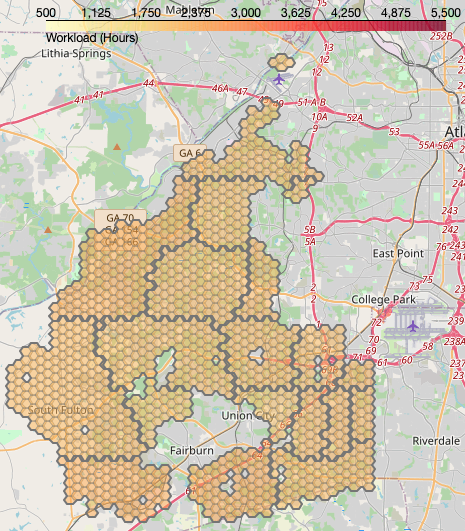}
\caption{Striping+SA(.1)}
\end{subfigure}
\hspace{.1in}
\begin{subfigure}[h]{0.22\linewidth}
\includegraphics[width=\linewidth]{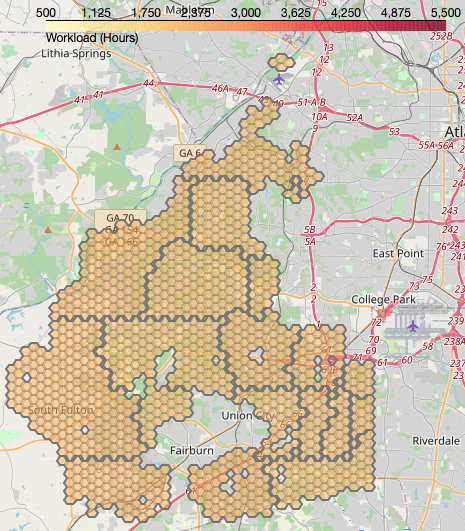}
\caption{Striping+SA(.05)}
\end{subfigure}
\caption{Districting results of South Fulton police beats. The performance of these partitions is listed in Table~\ref{tab:computations-police-deterministic}. The color depth represents the workload per part. The dark lines outline the part boundaries. 
(a-d) are the partitions with 7 beats; 
(e-h) are the partitions with 15 beats; 
(a, e) are generated by weighted $k$-means where the workload is regarded as weight;
(b, f) are generated by the striping algorithm;
(c, d, g, h) are the refined results of simulated annealing by taking (b, f) as their warm start initialization, respectively, with balance parameter $\varepsilon=.1$ in (c, g) and $\varepsilon=.05$ in (d, h).}
\label{fig:police-deterministic}
\end{figure}

\end{document}